\documentclass[3p,times,9pt]{myelsarticle}

\usepackage{lineno,hyperref}
\modulolinenumbers[5]
\usepackage{amsmath,amssymb,amsfonts}
\usepackage{algorithmic}
\usepackage{graphicx}
\usepackage{textcomp}
\usepackage{xcolor}
\usepackage{url}
\usepackage{mathtools}
\usepackage{subfig}
\usepackage{booktabs} 
\usepackage{tabularx}
\usepackage{pifont}

\let\norm\undefined 
\DeclarePairedDelimiter\norm{\lVert}{\rVert}
\DeclarePairedDelimiter\abs{\lvert}{\rvert}
\DeclarePairedDelimiter\taxicab{\lvert}{\rvert}%

\DeclareMathOperator*{\argmin}{\arg\!\min}

\def\wideparen#1{\mathop{\vbox{\ialign{##\crcr\noalign{\kern3\p@}
         \downparenthfill\crcr\noalign{\kern3\p@\nointerlineskip}
         $\hfil\displaystyle{#1}\hfil$\crcr}}}\limits}

\def\cents{\hbox{\rm\rlap/c}}

\newtheorem{thm}{Theorem}
\newtheorem{lem}{Lemma}

\newtheorem{cor}{Corollary}
\newproof{proof}{Proof}

\newtheorem{innercustomgeneric}{\customgenericname}
\providecommand{\customgenericname}{}
\newcommand{\newcustomtheorem}[2]{%
  \newenvironment{#1}[1]
  {%
   \renewcommand\customgenericname{#2}%
   \renewcommand\theinnercustomgeneric{##1}%
   \innercustomgeneric
  }
  {\endinnercustomgeneric}
}
\newcustomtheorem{customthm}{Theorem}

\newcommand{\mlless}{\textsc{MLLess}}

\newenvironment{itquote}
  {\begin{quote}\itshape}
  {\end{quote}\ignorespacesafterend}

\begin{document}

\begin{frontmatter}

\title{\mlless: Achieving Cost Efficiency in Serverless Machine Learning Training}

\author[urv]{Pablo Gimeno Sarroca}
\ead{pablo.gimeno@urv.cat}

\author[urv]{Marc S\'{a}nchez-Artigas\corref{correspondingauthor}}
\cortext[correspondingauthor]{Corresponding author}
\ead{marc.sanchez@urv.cat}

\address[urv]{Computer Science and Maths, Universitat Rovira i Virgili (Spain)}

\begin{abstract}
Function-as-a-Service (FaaS) has raised a growing interest in how to 
``tame'' serverless computing to enable domain-specific use cases such as data-intensive applications and machine learning (ML), to name a few. Recently, several systems have been implemented for 
training ML models. Certainly, these research articles are significant steps 
in the correct direction. However, they do not completely answer the nagging question of when
serverless ML training can be more cost-effective compared to traditional ``serverful''
computing. To help in this endeavor, we 
propose \mlless, a FaaS-based ML training prototype built atop
IBM Cloud Functions. To boost cost-efficiency, \mlless\phantom{} 
implements two innovative optimizations tailored to the traits of serverless computing: on one hand, a significance filter, to make indirect
communication more effective,  and on the other hand, a scale-in auto-tuner, to reduce cost by benefiting from the FaaS
sub-second billing model (often per $100$ms). Our results certify that \mlless~can be $15$X faster than 
 serverful ML systems~\cite{pytorch} at a lower cost for sparse ML models that exhibit fast convergence such as sparse logistic regression and matrix
factorization. Furthermore, our results show that \mlless~can easily scale out to increasingly large fleets~of serverless workers.
\end{abstract}

\begin{keyword}
Serverless computing \sep Function-as-a-Service \sep Machine Learning
\end{keyword}

\end{frontmatter}

\section{Introduction}

 A vivid interest has recently arisen over the issue of serverless computing and
its implications for general-purpose computations. Originally geared towards web microservices
and IoT applications, recently researchers have started to examine its potential in data-intensive applications~\cite{openlambda, pywren, gg, numpywren, crucial, pywrenibm, multicloud, lithops}.
Altogether, these works have led to a clear identification of  ``what'' workloads are best suited to serverless computing. 

Similarly, a recent trend on building machine learning (ML) on top of Function-as-a-Service (FaaS)
platforms~has emerged as a new research area~\cite{siren, cirrus, lambdaml, stratum, vatche18, silva}. 
Since  ML inference is a trivial use case of FaaS computing~\cite{stratum, vatche18}, attention has turned into 
ML model training, which is a deal more difficult. Despite all the preceding efforts, it still remains uncertain
under what conditions ML tranining on top of  FaaS may be beneficial. This is not a trivial question, as the 
evaluation of serverless ML training is not as simple as running VM-based ML systems such as PyTorch or TensorFlow
on top of~cloud functions. The fundamental reason is that traditional ML systems  have not been prepared to deal with
the idiosyncrasies of the FaaS model such as the impossibility of function-to-function
communication, the limited memory and transient nature of  serverless functions \cite{berkeley, cirrus}.

Our aim in this work is to understand the feasibility of supporting distributed ML training 
over FaaS platforms. Concretely, we are interested in the following question:
\begin{itquote}
When can a FaaS platform be more cost-efficient than
a VM-based, ``serverful'' substrate  (IaaS) for distributed ML training?
\end{itquote}

To help in this endeavor, we introduce \mlless, a prototype FaaS-based ML training system atop
IBM Cloud Functions. To pick a point in the design space that is more cost-efficient than
the prior serverless ML systems~\cite{siren, cirrus, lambdaml}, \textsc{MLLess} comes up with two novel optimizations tailored 
to the traits of the FaaS model. Our view is that in the same way that serverful ML training has been specialized for coarse-grained VM-based
clusters, a fair comparison between FaaS and IaaS is not possible unless model training is specialized to address the limitations of the FaaS computing model.
Our two novel optimizations pursue this noble goal. The first optimization reduces the bandwidth requirements of exchanging model updates between
workers using shared external storage, yet assuring convergence. The rationale behind this optimization is the ``stateless'' essence of FaaS, which does not allow
concurrent functions to directly share state. Thus, any model update (e.g., a gradient) must be exchanged through remote storage.

The second specialization is a scale-in auto-tuner to 
increasingly decrease the number of workers, so as the cost of training, with no side effects on 
convergence. This method benefits from the ``pay-per-usage'' cost 
model of FaaS to save money, instead of the reservation-based model that charges end users for idle VM resources. 
From an ML perspective, FaaS thus promises more savings, since only the active workers at any given time will be billed,
bringing out a superior cost-efficiency than IaaS if the pool of workers is optimally shrunk during model training.

Equipped with this specialized training architecture,  we next use \mlless~to investigate the cost-efficiency of FaaS for
ML training. Since the per-minute cost of executing a cloud function is higher than its resource-equivalent VM instance (see Table~\ref{tab:pricing}), 
we focus on \emph{fast-convergent models} here, which intuitively are the most amenable to serverless computing. 
ML models that take hours to converge are presumably more cost-optimal to be trained on VM instances with
today's offerings. We also examine the scalability of \mlless, and compare the effectiveness of our significance filter to
loose synchronization models such as the Stale Synchronous Parallel (SSP)~\cite{ssp}. This is comparison is very
interesting, since while SSP restricts how stale, or ``old'', a model parameter can be, our significant filter bounds how
inaccurate a parameter can be. Hence, shedding light on what type of synchronization strategy is more appropriate  
for the indirect communication model of serverless computing is of vital importance.

\medskip
\noindent\textbf{Main insights.} Our study yields three key insights:
\begin{enumerate}
\item \emph{FaaS can be more cost-efficient than ``serverful'' libraries} such as PyTorch~\cite{pytorch} for models that quickly converge. Although
indirect communication severely penalizes FaaS-based ML training, \mlless~ameliorates its impact
with the aid of its two main optimizations, being $15$X faster while $6.3$X cheaper than PyTorch. It must be noted that for dense models, though, the benefits
of \mlless~are narrower, which suggests that where FaaS-based ML training excels is in \emph{sparse} models.
\item \emph{Specializing distributed ML training to FaaS is crucial~to yield a higher cost-efficiency.}
This requires dealing with low-level issues such as high gradient sparsity, filtering out non-significant updates, or
dynamically scaling down the pool of workers, i.e., abilities that are not always available in VM-based ML systems such as PyTorch. 
\item \emph{Filtering non-significant updates is better than bounded staleness for FaaS-based ML training.} While SSP has proven to
be very effective for distributed ``serverful'' ML training, it renders only a marginal benefit for model training over FaaS. Since the communication
of updates is the major bottleneck in FaaS, the flexibility to delay the propagation of an update until it
eventually becomes significant is more effective than tolerating some amount of staleness. 
\end{enumerate}
\medskip

\noindent \textbf{Reproducibility and open source artifacts.} \mlless~is publicly available at \url{https://github.com/pablogs98/MLLess}.

\medskip
\noindent \textbf{Roadmap.} A preliminary version of this work has appeared at Middleware'21~\cite{mlless21}. The rest of the article is structured as follows: \S\ref{sec:2}
discusses the challenges of FaaS for ML training. \S\ref{sec:mlless} presents \mlless' design. \S\ref{sec:opt}
details \mlless' major optimizations. \S\ref{sec:impl} gives implementations details, and \S\ref{sec:eval}
presents experimental results. \S\ref{sec:rw}surveys related work, and \S\ref{sec:end} concludes.

\section{Is FaaS Appropriate for ML Training?}
\label{sec:2}

Although the main innovation of serverless is hiding servers, what makes
serverless computing so powerful for training models is: 
\begin{itemize}
\item A ``pay-as-you-go'' model that does not charge users~for idle resources; and
\item Rapid and unlimited scaling up and down of resources to zero if necessary.
\end{itemize}
By playing out with the two essential qualities to a greater~or lesser extent, 
state-of-the-art FaaS-based ML systems~\cite{siren, cirrus, lambdaml} have
inadvertently established a rich design space in their attempt to circumvent 
the stringent limitations of the FaaS model. In our case, we leverage these two 
properties to our favor to design our scale-in auto-tuner (\S\ref{sec:sched}).
Moreover, we take the other tack and optimize indirect communication, as
it is the primary training bottleneck  (\S\ref{sec:sf}). Altogether, these
two forces, namely auto-scalability and general performance optimization,  
have enabled us to show that FaaS can be more cost-efficient than ``serverful'' computing (IaaS).

\medskip
\noindent \textbf{Limitations.} Despite the good news, it is important to remind  ourselves about
the most prominent hurdles to serverless ML training. First, 
today's FaaS platforms only support stateless function calls 
with limited resources and duration. For instance, a function call in 
IBM Cloud Functions can use up to $2$GB of RAM and must finish 
within $10$ minutes\footnote{\url{https://cloud.ibm.com/docs/openwhisk?topic=openwhisk-limits}}. 
Such limits automatically discard some natural practices such as
loading all training data into local memory, which must be downloaded remotely from shared storage in mini-batches, 
while inhibiting the use of any ML library that
has not been designed with these constraints in mind. For instance, the authors of Cirrus~\cite{cirrus}, a serverless ML
system, found impossible to run Tensorflow~\cite{tensorflow} or Spark~\cite{spark} on AWS lambdas with
such resource-constrained setups.

However, the most
critical issue is the impossibility of direct communication, which requires
a trip through shared external storage to pass state between functions. 
This not only contributes significant extra latency, often hundreds of milliseconds, but also
prevents exploiting HPC communication topologies adopted in ML 
such as tree-structured and ring-structured \emph{all-reduce}~\cite{ring}. For this reason,
a careful optimization of communication, ranging from serialization of high sparsity models 
to the development of communication-reduction techniques, such as our significance filter,  is crucial. 

\section{\mlless}
\label{sec:mlless}

We implement \mlless, a prototype FaaS-based ML training system 
built on top of IBM Cloud Functions. In this section, we describe its main
components and defer the explanation of our two key optimizations
to \S\ref{sec:opt}.

\begin{figure}
\centering
\includegraphics[width=0.75\textwidth]{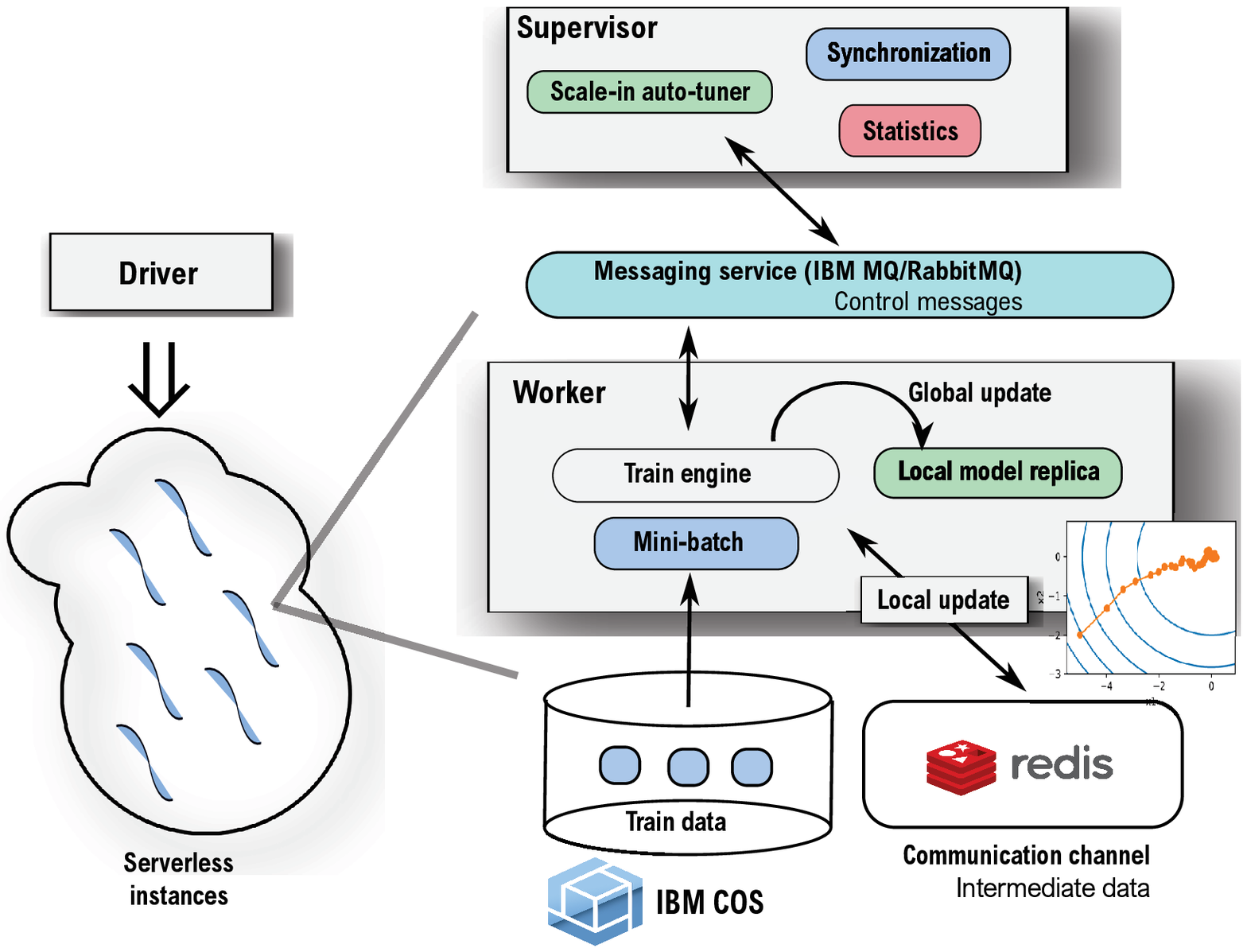}
\caption{\textbf{\mlless~system architecture.}}
\label{fig:1}
\vspace{-10pt}
\end{figure}

\subsection{System Overview}

An architectural overview of \mlless~is illustrated in Fig.~\ref{fig:1}. 
\mlless\hphantom{} consists of a \emph{driver} that runs on the local machine of the  
data scientist. When the user launches a ML training job, the driver 
invokes the requested number of serverless \emph{workers}, who execute the 
job in a data-parallel manner. Each worker maintains a \emph{local replica of the model}
and uses the library of \mlless~to train it. We have chosen this decentralized design for \mlless~since it better abides by to a pure
FaaS architecture compared to the \emph{VM-based parameter server} \cite{parameter} model, e.g., followed by 
other works such as Cirrus~\cite{cirrus}.
\medskip

\noindent \textbf{Supervisor.} Since the driver is typically far from the data center (e.g., at a university lab), tasks, such as aggregating statistics to find whether the convergence criterion has been reached,
can introduce significant delays. To minimize latency, the driver also starts up a serverless function which acts as a supervisor. The
role of the supervisor is to collect and aggregate statistics, synchronize worker progress, e.g., in order to bound the divergence between model copies,
and terminate the training job when the stopping criteria is fulfilled, among other tasks. Nevertheless, one of the 
core attributions of the supervisor is to automatically remove workers when their marginal contribution to convergence is minor, or even negative
due to increased communication costs (please, see \S\ref{sec:sched} for details). 

Since the supervisor is a serverless function, it is subjected to the time constraints of the underlying FaaS platform, in this case, to a
maximum execution time of $600$ seconds (IBM Cloud Functions). Although the supervisor never ran out of time in our experiments,
it would not be laborious for the supervisor to pause execution when the $10$-minute timeout is close, checkpoint its internal state to storage and
re-launch it as a new worker.
\medskip

\noindent \textbf{Synchronization.}  The iterative nature of ML algorithms 
may imply certain dependencies across successive iterations. To keep consistency, synchronizations 
between workers must happen at certain boundary points. To this aim, \mlless\hphantom{} supports
different consistency models: the Bulk Synchronous Parallel (BSP) model where the workers must wait for each other at the end of every
iteration, and the Stale Synchronous Parallel (SSP)~\cite{ssp}, a synchronization model that relaxes consistency by 
permitting workers to read stale parameter values as long as they are not too ``stale''.
SSP was proposed to overcome the \emph{straggler problem} suffered by BSP, where each iteration proceeds at the pace of the slowest worker.
For this reason, SSP defines an explicit ``slack'' parameter for coordinating progress among the workers. The slack
specifies how many iterations out-of-date the local replica of a worker can be, which implicitly dictates how far ahead of the slowest
worker any worker is allowed to progress. For instance, with a slack of $s$, a worker at iteration $t$ is guaranteed to see
all updates from iterations $1$ to $t - s - 1$, and it may see (not guaranteed) the updates from iterations $t - s$ to $t - 1$.
We set BSP as the default synchronization model because  it simplifies the reasoning about the impact of our optimizations
on model convergence.

\mlless\hphantom{} also includes a variant of BSP where the workers only send those updates that are significant. This
variant reduces communication costs, but allows local model copies to diverge across workers (see \S\ref{sec:sf}). Compared~with
SSP,  our variant restricts how inaccurate the aggregated update for a model parameter can be, in comparison to its current value,
instead of bounding stateleness  in terms of the iteration count.
\medskip

\noindent \textbf{Communication channels.}  Due to the absence of direct communication
between the workers or with the supervisor, \mlless~establishes two channels of \emph{indirect communication:} 
\begin{itemize}
\item \textbf{Signaling channel.} For exchanging  control messages between the workers and the supervisor (e.g., to signal a worker to advance to the next iteration), it
leverages a messaging service built on RabbitMQ\footnote{\url{https://www.rabbitmq.com}}, 
though it could be replaced by the native IBM's MQ messaging service\footnote{\url{https://www.ibm.com/products/mq}} 
without complications.
\item \textbf{Intermediate state.} For sharing the intermediate state generated
during model training (e.g., local gradients), \mlless~employs Redis\footnote{At the time of writing this paper, there is no serverless cache service in IBM Cloud, so users
still need to provision cache instances themselves.}, a~low-latency, in-memory key-value store that supports thousands of requests/s
\cite{locus}.
\end{itemize}

~To store the input dataset mini-batches, \mlless~uses IBM COS, the serverless object storage service from IBM Cloud. 
Since it is  an ``always on'' service, it does not incur any startup delay, though it is a little bit slower compared to Redis,
but sufficient for our purposes. 

\subsection{Model Training}

\mlless~assumes that a training dataset $D$ consists of $N$  independent and identically distributed (IID) data samples drawn 
by the underlying data distribution $\mathcal D$.
Let $D = \left\{(\mathbf{v}_i \in \mathbb{R}^n, l_i \in \mathbb{R}) \right\}^N_{i=1}$,
where $\mathbf{v}_i$ denotes the feature vector and $l_i$ represents the label of the 
$i^\textrm{th}$ data sample. The objective of training is to find an ML model $\mathbf{x}$ 
that minimizes  a loss function $f$ over the dataset $D$: $\argmin_ \mathbf{x} \frac{1}{N} \sum_i  f(\mathbf{v}_i, l_i, \mathbf{x})$.

In today's systems, one typical optimizer is Stochastic Gradient Descent (SGD)~\cite{sgd}, 
 an iterative training algorithm that adjusts $\mathbf{x}$ based on a few samples at a time:
\[
\mathbf{x}_t=\mathbf{x}_{t-1} - \eta_t\nabla f_{\mathcal B_t}(\mathbf{x}_t),
\]

where $\eta_t$ is the learning rate at step $t$ of the algorithm, $\mathcal B_t$ is a mini-batch of $B$ training samples, and $\nabla f_{\mathcal B_t}(\mathbf{x}_t)$ is the gradient of
the loss function, averaged over the batch samples: $\nabla f_{\mathcal B_t}(\mathbf{x}_t) = \frac{1}{B} \sum_{(\textbf{v}, l) \in \mathcal B_t} \nabla f\left(\mathbf{v}, l, \mathbf{x}_t\right)$.
\mlless~supports different optimizers (see Table~\ref{tab:eval} for further details).

Since serverless workers have very limited memory, e.g., IBM Cloud Functions 
can only access at most $2$~GB of local RAM, it is infeasible to replicate all
training data into memory. Hence, \mlless~assumes that 
the training dataset is stored in an object store, i.e., IBM COS, and partitioned into
mini-batches of size $B$. To generate the mini-batches in the appropriate format (e.g., 
feature normalization), 
\mlless\hphantom{} leverages PyWren-IBM~\cite{pywrenibm}, a FaaS-based map-reduce framework. 
For instance, by chaining two map-reduce jobs, it
is straightforward to normalize a dataset using \texttt{min-max} scaling, where the first map-reduce job 
gets the minimum and maximum values of each feature, and the second one does
the actual scaling.
\medskip

\noindent \textbf{Job execution.} A typical execution of a training job with \mlless~involves the following three steps: \ding{182} Once up and running, each worker creates
a local copy of the model with the aid of the \mlless~library, and starts to optimize
 the loss function $f$; \ding{183} In each iteration, eachworker separately fetches
a mini-batch from IBM COS, and then it calculates a local update from its model replica before synchronization takes place.
The type of local update depends on the ML algorithm. In the case of the Stochastic Gradient Descent (SGD)~\cite{sgd} 
algorithm, local gradients are averaged to obtain a global gradient update; \ding{184} Due to the lack of direct communication,
each worker independently of the others pulls all the local updates from external storage (Redis), and aggregates them
to update its local model copy. The availability of a local update is announced to the rest of workers through the
signaling channel. 

It is worth to note here that this decentralized design is easy to scale out, since no single component is
responsible for merging all the local updates,  as it occurs in LambdaML~\cite{lambdaml}. Indeed, to
scale to large input data sizes or number of workers, it suffices to add more Redis instances and shard the local, ephemeral updates
from the workers across~them, so that the load is evenly distributed among all Redis shards. Further, the separation of control and
data flows makes it easy to support different consistency models, from strict models such as BSP to relaxed ones such as SSP, with
little or no changes in the iterative optimizers. 
\medskip 

\noindent \textbf{Weak scaling.} \mlless~parallelism strategy keeps the mini-batch
size $B$ the same when the number of worker reduces by the action of our scale-in auto-tuner (see \S\ref{sec:sched}).
The reason is to avoid that every change in the number of workers incurs costly data repartitioning transfers to adjust
the mini-batch size, since it may lower the net benefit of worker downscaling. Nonetheless, this entails that the
global batch size $B_\textrm{g}$ decreases linearly with the number of workers $P$. That is, $B_\textrm{g} = PB$, which may affect
the convergence speed of the optimizer~\cite{weak}. To prevent significant deviation, the auto-tuner only removes  a worker if the degradation in loss reduction 
does not exceed a certain threshold (see \S\ref{sec:sched} for details).

\section{Optimizations}
\label{sec:opt}

FaaS is typically more expensive in terms of \$ per CPU cycle than ``serverful''
computing. This means that a priori, a user optimizing for cost would
likely prefer IaaS over FaaS. Fortunately, FaaS-based training runtimes
still show a large margin of improvement that can lead to more cost-effective 
training, particularly, for models that converge fast. Here~we describe two
optimizations to confirm this intuition. In \S\ref{sec:sf}, we elaborate on
an optimization to improve throughput, and discuss the scale-in autotuner
details in \S\ref{sec:sched}.

\subsection{Significance Filter}
\label{sec:sf}

As cloud providers disallow direct communications between 
functions, fast aggregation of gradients cannot be made with
optimal primitives such as ring \texttt{all-reduce}~\cite{ring}, and
must be done through external storage.  Despite \mlless\vphantom{} uses a
low-latency key-value store such as Redis for this purpose, 
the exchange of updates is, as expected, a high-cost operation, which
can significantly diminish the benefits of parallelism. This is
particularly visible for the Bulk Synchronous Parallel (BSP) model of computation, 
where no worker can proceed to the next step without having all
workers finish the current step.

To reduce strain on external storage, \mlless~comes along with a variant
of the Approximate Synchronous Parallel (ASP) model~\cite{gaia},
we name it \emph{`Insignificance-bounded Synchronous Parallel' (ISP)}
to distinguish it from the original consistency model. In short, ASP was originally proposed to break the communication bottleneck
over WANs in~geo-distributed ML systems. The central idea of ASP was
to remove insignificant communication across data centers, yet
ensuring the correctness of ML algorithms. 
In this sense, ISP borrows from ASP the idea of filtering non-significant
updates, but applies it to accelerate the broadcast of local gradients 
between workers \emph{within the same data center, or cluster}. Since ISP operates at the cluster level, its implementation is much
simpler than ASP. It does not need complex synchronization mechanisms between data
centers such as the ASP selective barrier and mirror~\cite{gaia}, which facilitates its adoption in current serverless architectures. Further,
ASP was originally implemented using the parameter server model~\cite{parameter}, and not for fully decentralized training systems
such as \mlless.
\medskip

\noindent \textbf{Overview.}
In a nutshell, ISP can be viewed as a technique to reduce the per-step communication complexity while preserving 
the convergence rate, which results in an improvement in system throughput, so as job training times. More
concretely, its goal is to reduce the size of the local update to be transferred to the rest of workers, 
after the local worker goes  through its mini-batch. This is possible because ISP benefits from the 
robustness of many ML algorithms (e.g., logistic regression,
matrix factorization, collapsed Gibbs, etc.), which tolerate a bounded amount 
of inconsistency. To ensure equal algorithmic progress per-step, ISP 
enables users to tune the strictness of the significance filter to achieve the sweet spot. 
Typically, the strictness is controlled by a threshold $v$, 
which is reduced over time. That is, if the initial threshold is $v$, then the 
threshold value $v_t$ at step $t$ of the ML algorithm is given by $v_t = \frac{v}{\sqrt{t}}$.

Very importantly, ISP is synchronous in nature. That is,  all  workers must finish the current step before proceeding
to the next iteration. Therefore, ISP differs from bounded asynchronous consistency models such as SSP~\cite{ssp},
which sets a fixed upper-bound on the iteration gap between the  fastest worker and the slowest one. The focus of
ISP is thus on reducing communication requirements rather than alleviating system heterogeneity as SSP realizes. Observe that
a smaller communication complexity also means a lower computation complexity, since less model parameters must
be updated per iteration.

It is worth to note here that the original definition of the ASP model~\cite{gaia} 
does not presume a specific significance function, as clearly reflected in its proof of 
convergence in Theorem 1, which only provides a general analysis of ASP.
However, to yield a more robust evidence of ISP validity, we ``incarnate'' the ISP model with a concrete
significance filter, and prove its convergence exclusively for this function.
\medskip

\noindent \textbf{Significance function.} To trim communication while
bounding deviation between any two model replicas, a ``clever'' compression technique
is to have each worker aggregate its local updates while they are non-significant. In this way,
if the accumulated update eventually becomes significant, the worker will be able to broadcast the complete history 
of its non-significant updates encoded as a single update,~thereby minimizing both the communication burden and
deviation from the ``true'' mini-batch gradient. 

More formally, let $\mathbf{x}_t \in \mathbb{R}^n$ be the
parameters of the model at step $t$, and $\mathbf{u}_t$ be the associated update s.t. 
$\mathbf{x}_t=\mathbf{x}_{t-1} + \mathbf{u}_t$. As  the  update  operation  is  associative  and  commutative,  we
simply aggregate the non-significant updates for any model parameter 
by summing  them up. Eventually, the \emph{per-parameter}
accumulated update may become significant and be pushed to the rest of workers.
Let $t_{p_i}$ be the last propagation time for the $i^\mathrm{th}$ parameter.
Then, we define the per-parameter significance filter as:
$\left| \frac{\sum^t_{t^\prime = t_{p_i}} u_{i, t^\prime}}{x_{i, t}}  \right| > v_t.$

\noindent Note that with the above significance filter, the compression
factor becomes proportional to the number of accumulated updates, i.e., $m_t := \left(t-t_{p_i}\right)$.
This number can be arbitrarily big, provided that the magnitude of the accumulated update relative to the 
current model parameter value is less than $v_t$. For this reason, it is key to show that ISP is able to
 maintain an approximately-correct copy of the global model in each worker. We formulate this
in Theorem~\ref{thm:1}:

\begin{thm}
\label{thm:1}
Suppose we want to find the minimizer $\mathbf{x}^*$ of a convex function $f(\mathbf{x}) =\sum^T_{t=1} f_t(\mathbf{x})$ (components $f_t$ are also convex) via
SGD on one component $\nabla f_t$ at a time. Also, the algorithm is replicated across $P$ workers with synchronization at every step $t$. Let
$\mathbf{u}_t := - \eta_t \nabla f_t(\widetilde{\mathbf{x}}_t)$, where the step size $\eta_t$ decreases as $\eta_t=\frac{\eta}{\sqrt{t}}$. As per-parameter significance
filter, we use $\left| \frac{\delta_{i, t}}{\widetilde{x}_{i, t}}  \right| > v_t$,  where $\widetilde{x}_{i, t}$ is the $i^\mathrm{th}$ parameter of the noisy
state $\widetilde{\mathbf{x}}_t := (\widetilde{x}_{0, t}, \widetilde{x}_{1, t}, \dots, \widetilde{x}_{n, t})$ at step $t$, $\delta_{i, t} := \sum^t_{t^\prime = t_{p_i}} u_{i, t^\prime}$ denotes the
accumulated update for the $i^\mathrm{th}$ parameter since the last propagation time $t_{p_i}$, and $v_t$ is the significance threshold that decreases as $v_t=\frac{v}{\sqrt{t}}$.
Then, under suitable conditions: $f_t$ are $L$-Lipschitz and the distance between any $\mathbf{x}$, $\mathbf{x}^\prime$ in the parameter space
$D(\mathbf{x}, \mathbf{x^\prime}) \leq \Delta^2$ for some constant $\Delta$:
\[
R[X] := \sum^T_{t=1} f_t(\widetilde{\mathbf{x}}_t) - f_t(\mathbf{x}^*) = \mathcal{O}\left(\sqrt{T}\right), 
\]
and thus $\lim_{T \to \infty}\frac{R[X]}{T} = 0$.
\end{thm}

We provide the details of the proof of Theorem~\ref{thm:1} and the notations in~\ref{sec:analysis}.

\subsection{Scale-in Scheduler}
\label{sec:sched}

Compared with cluster computing, one major advantage of the FaaS 
model is that it enables the rapid adjustment~of the number of workers over time.
For instance, the removal of a worker in the middle of a training job does 
not leave cluster resources unallocated, or demand a prompt re-allocation of them to
other concurrent jobs such as in the case of reserved VMs~\cite{slaq, optimus}. This ability opens the door to the
invention of novel schedulers that, for example, minimize monetary cost by  dynamically
adjusting the number of workers as the job progresses. This may result in a more~cost-effective training, compared to traditional ``serverful'' cloud computing, which charge
customers based on the time that the reserved VMs remain active.

\begin{figure}
\centering
\includegraphics[width=0.4\textwidth]{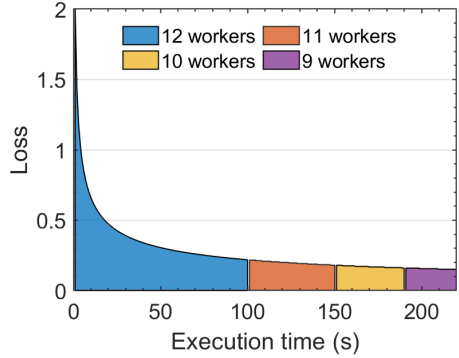}
\caption{Sample execution of the scale-in scheduler. In the low convergence zone, workers are increasingly eliminated from the pool.}
\label{fig:sched}
\vspace{-5pt}
\end{figure}

\begin{figure}
\centering
\subfloat[][\textbf{Training speed.}]{
  \includegraphics[width=0.33\textwidth]{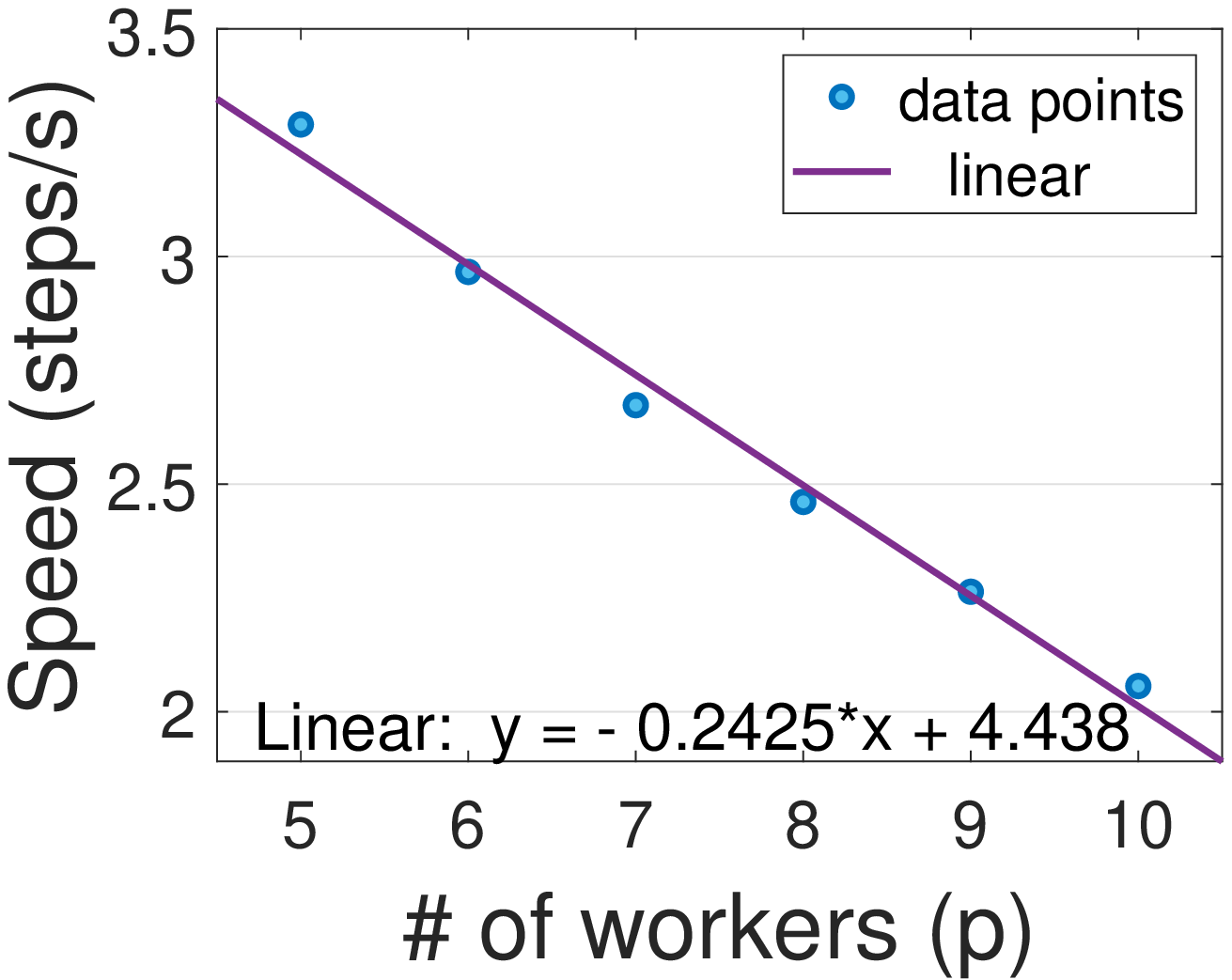}
	\label{fig:2}
}\quad
\subfloat[][\textbf{Reference curve fitting:} {\small $\theta_0 = 0.05$, $\theta_1 = 1.58$, $\theta_2 = 0.58$, $\theta_3 = 0.49$}.]{
  \includegraphics[width=0.33\textwidth]{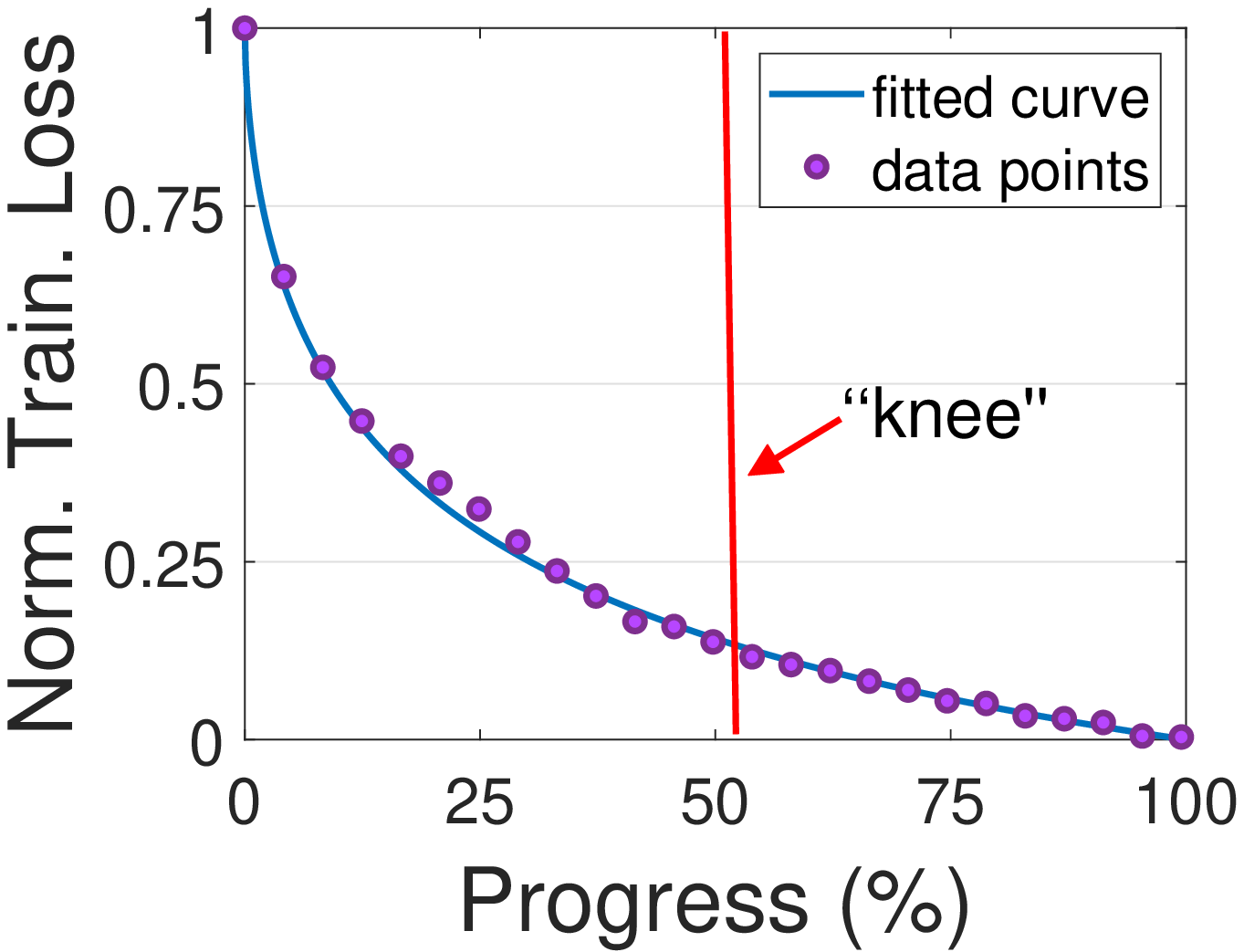}
	\label{fig:3}
}\\
\subfloat[][\textbf{Prediction error} in estimation of $50$-$200$ steps in advance.]{
  \includegraphics[width=0.33\textwidth]{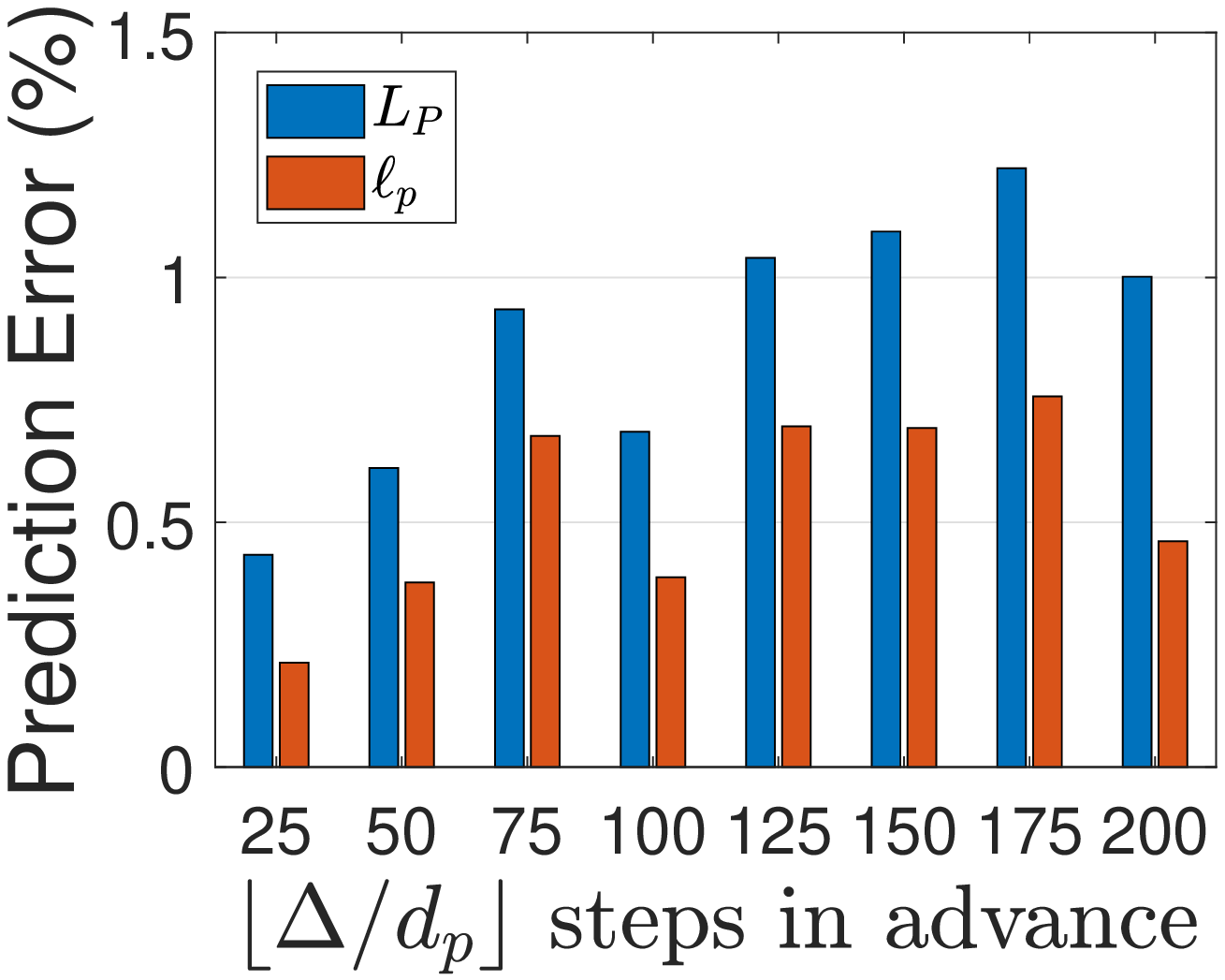}
	\label{fig:4}
}\quad
\subfloat[][\textbf{Prediction error} of the model $\ell_p(t)$ as job progresses.]{
  \includegraphics[width=0.33\textwidth]{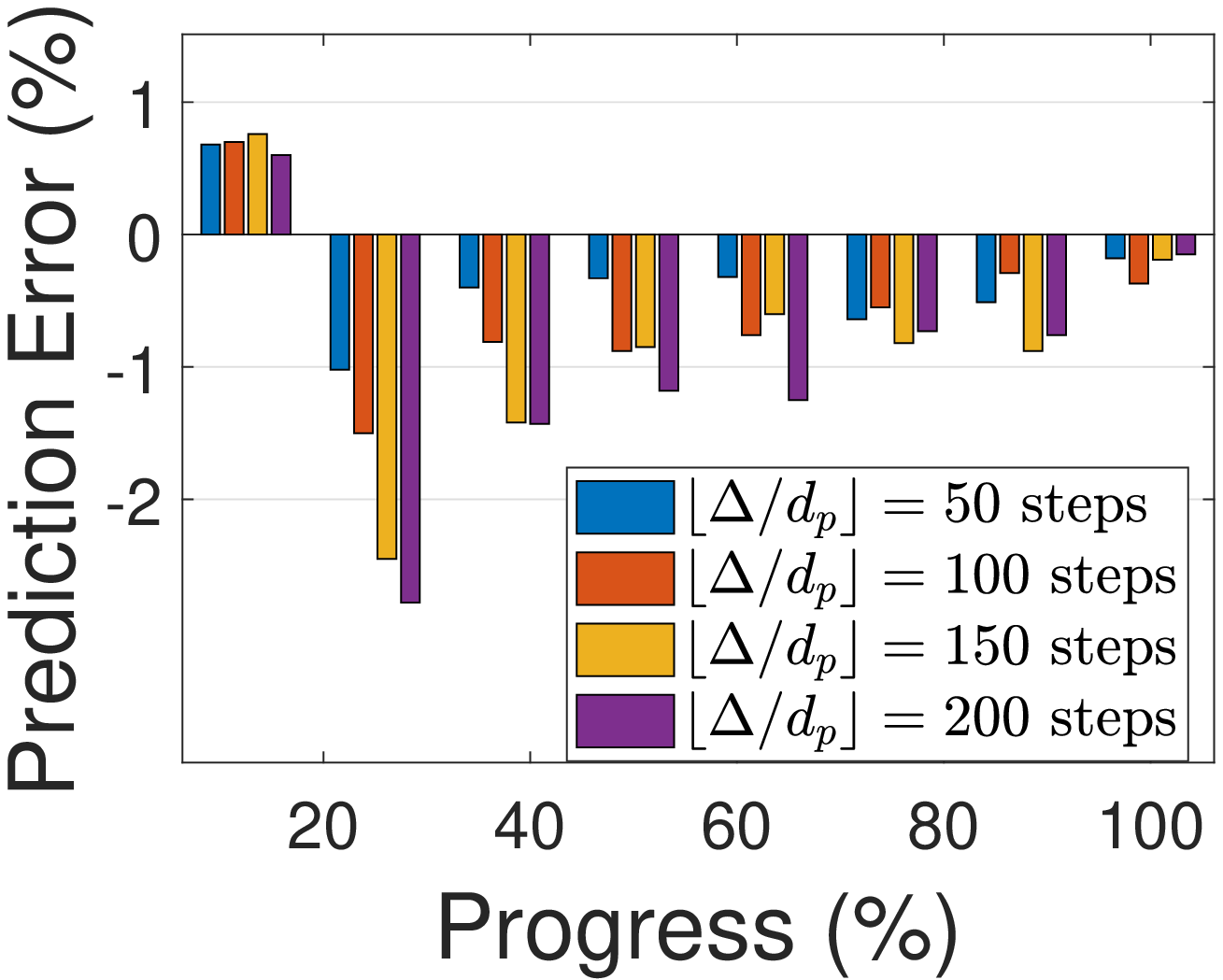}
	\label{fig:5}
}
\caption{Training speed and prediction error of training a matrix factorization (PMF) model~\cite{pmf} on MovieLens-1M data~\cite{movielens}.}
\vspace{-15pt}
\end{figure}

To show that a better cost-efficiency ratio is possible with FaaS computing,
\mlless~includes a dynamic and fine-grained scheduler designed to
remove ``unneeded'' workers. ML training is typically an iterative process where
the level of quality improvement decreases as the number of training steps increases~\cite{slaq}.
For instance, SGD reduces loss  approximately as a geometric series on convex problems~\cite{convex}. 
This implies that, while a higher number of workers is desirable during the first training steps to steeply diminish
loss, a large worker pool gives only marginal returns when loss reduction slows down, which 
ends ups worsening the cost-efficiency ratio. In this sense, the primary objective of \mlless's scale-in scheduler
is to increasingly cut down the training cost as the job progresses in order to maintain cost-effectiveness.
Fig.~\ref{fig:sched} depicts an example of how workers are increasingly removed from the system when loss
reduction stagnates as a result of the action of the scale-in scheduler.
\medskip

\noindent \textbf{Algorithm.}
 From an initial number of workers $P$, the scale-in scheduler dynamically reduces 
the worker pool based on the feedback of the ML algorithm, which includes not~only the loss values but
also the speed of the training steps. Using the loss information, the scheduler first detects the ``knee'' in
the convergence rate, after which loss reduction slows down significantly, and uses the history
of loss values at this time to fit the reference training loss curve $L_P(t)$. This curve will be used
by the scheduler to quantify the deviation from the original convergence rate introduced by a future
removal of a worker. Further, the scheduler estimates the reference step duration $d_P$ by averaging the duration
of all training steps up to this time. 

After estimation of these quantities, the scheduler removes the worker with the lowest-quality replica of the model from the pool, and waits for~the next scheduling interval. 
Now let $1 < p \leq P - 1$ denote the current number of workers. Then, the scheduler repeats the following sequence of operations upon each scheduling interval:
\begin{enumerate}
\item \emph{Estimation phase}. It fits a new training loss curve $\ell_p(t)$. But, at this time, it uses only the loss values
collected so far since the last worker removal. The key reason is that the removal~of a worker may
affect convergence due to weak scaling~\cite{weak}, for it is required a new fitting to
capture the potential deviation from the reference curve. Also, it estimates the current step duration~$d_p$
by the same procedure as above. Computation~of this estimate is necessary as $d_p < d_P$. This~occurs
because the per-step communication overhead is ~$\mathcal{\widetilde{O}}(p)$, 
where $\mathcal{\widetilde{O}}$ hides the dependence on the model 
size. This is easy to see in Fig.~\ref{fig:2}, where 
a matrix factorization model~is trained with a varying number of workers. The figure
shows  how training speed decreases linearly with the number of workers. 
As we fix the local mini-batch size to avoid repartitioning data,
less workers implies less data to  pull from external storage per iteration, so
as the communication overhead.
 \item \emph{Decision phase}. In this phase, the scheduler decides to remove a
new worker based on the relative error in the projected~loss reduction 
in time horizon $\Delta$:
\begin{equation}
\label{eq:re}
s_\Delta(t) := \left[ \frac{L_P\left(t + \left\lfloor \frac{\Delta}{d_P} \right\rfloor\right) - \ell_p\left(t + \left\lfloor \frac{\Delta}{d_p} \right\rfloor\right)}{L_P\left(t + \left\lfloor \frac{\Delta}{d_P} \right\rfloor\right)} \right], 
\end{equation}
where:
\begin{align*}
	t &:= \textrm{~current training step}, \\
	 \left\lfloor \frac{\Delta}{d_p} \right\rfloor &:= 
	 \begin{array}{l}
		\textrm{no. of steps to be completed in $\Delta$ time} \\
		\textrm{ units with $p$ workers},
	 \end{array} \\
	 	L_P\left(t + \left\lfloor \frac{\Delta}{d_P} \right\rfloor\right) &:= \textrm{~expected loss with all $P$ workers}, \\
	\ell_p\left(t + \left\lfloor \frac{\Delta}{d_p} \right\rfloor\right) &:= \textrm{~expected loss with the $p$ workers}, \\
\end{align*}
Then, the scaling-down condition is simply whether~this term is below a certain threshold: $s_\Delta(t) < S, S \in [0, 1]$. 
Intuitively, this term tells how much the convergence rate of the ML algorithm may worsen with $p$ workers compared to the 
original $P$-worker configuration~in~the region of slow convergence.
Note that the value of~$s_\Delta(t)$ can be negative, which
means that system throughput is indeed better as a result of removing workers.~This can happen if the decrease
in the communication cost outweighs the loss of parallelism, for instance. 
\end{enumerate}

Finally, we want to signal that although the parameter $\Delta$ can take arbitrary values, it has been designed to
anticipate the behavior of the system before a new scheduling interval arrives. Presume a fixed scheduling epoch 
of duration $T$. Since a new scheduling decision can be made after time $T$, the idea is to choose $\Delta \leq T$ to
ascertain whether the removal of a worker is beneficial in a short time horizon $T$. In general terms, the value of $\Delta$
will vary depending upon the specific ML job. The reason is  that while iterations may last $10$-$100$ ms in some ML jobs,
they may take a few seconds to complete in others. Irrespective of the ML algorithm, performing scheduling on short intervals  
could be disproportionally expensive due to the scheduling overhead, which involves function fitting in our case.
\medskip

\noindent \textbf{Loss deviation.} To predict how far a declining worker pool
may deviate from the initial convergence rate, as defined in Eq.~(\ref{eq:re}), the scheduler performs online
fitting on two types of learning curves, namely, the reference curve, $L_P(t)$, and the family of curves,
$\left\{\ell_p(t)\right\}_{1 < p \leq P - 1}$, drawn as the number of workers decreases over time.
To improve  prediction accuracy, each type of curve has a different shape for the following reason.
While $L_P(t)$ is built on the loss values from the region of fast convergence, 
the curves $\left\{\ell_p(t)\right\}_{1 < p \leq P - 1}$ are much more flat, as they
correspond to the region where loss reduction~slows down and stabilizes, so assuming
an appropriate curve for each region makes prediction more fine-grained.

We observe that most ML jobs use first-order algorithms such as mini-batch SGD~\footnote{Assume the loss function $f$ is convex,  differentiable,  and $\nabla f$ is Lipschitz continuous.},
which exhibits a convergence rate of $\mathcal O(1/\sqrt{Bt}+ 1/t)$~\cite{minibatchSGD}, where $B$ denotes the mini-batch size.
Consequently, we use the following model for the reference curve:
\begin{equation}
L_P(t) := \frac{1}{\theta_0 t^{\theta_1} + \theta_2} + \theta_3,
\end{equation}

where $\theta_0$, $\theta_1$, $\theta_2$ and $\theta_3$ are non-negative coefficients. 
An example of online curve fitting when training a PMF model is
depicted in Fig.~\ref{fig:3}. For the slow-convergence curves, we set:
\begin{equation}
\ell_p(t) := \frac{1}{\theta_0 t^2 + {\theta_1} t + \theta_2} + \theta_3,
\end{equation}

as in~\cite{slaq}, where $\theta_0$, $\theta_1$, $\theta_2$ and $\theta_3$ are also non-negative.  We
utilize a non-negative least squares solver~\cite{solver} to fit the points~in all the curves. Before
doing curve fitting, the loss values are always passed through an exponentially weighted moving average (EWMA) 
filter to remove outliers. 

Retaking the MF training example, Fig.~\ref{fig:4} gives
the error when estimating the loss values for an increasing
number of steps ahead from the ``knee''. Here the prediction error is
the difference between the actual and estimated loss values, divided by the actual one.
As shown in Fig.~\ref{fig:4}, both the reference curve $L_P(t)$ and the slow-convergence model $\ell_p(t)$ 
achieve a prediction error inferior to $1.5\%$, even when predicting up to $200$ steps in advance.
Finally, Fig.~\ref{fig:5} shows how estimation improves as more and more data points are collected 
for fitting the curve $\ell_p(t)$, irrespective of how many steps are predicted in advance.  
\medskip

\noindent \textbf{Automatic ``knee'' detection.} To favor convergence, the scale-in scheduler 
never eliminates a worker before passing the ``knee''. The reason is to maximize the time that the~ML algorithm stays 
within the region of fast convergence, only scaling down the number of workers once the learning curve starts to flatten out. 
There are several methods out there to automatically identify ``knee'' points from discrete data~(e.g.,  \cite{energy} and Kneedle~\cite{kneedle}),
which can be plugged into \mlless\phantom{} without further adaptations. For all ML jobs considered in this work, though, a simple threshold-based 
heuristic on the first derivative of the learning curve, i.e.,  the slope of the tangent line, worked well in all cases. 
\medskip

\begin{figure}
\centering
 \includegraphics[width=0.33\textwidth]{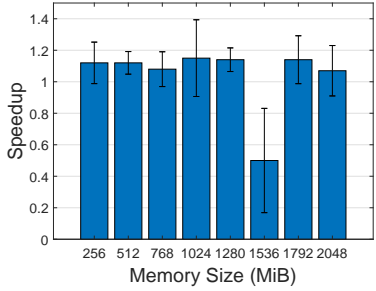}
\caption{Speedup of two threads relative to single-thread performance within a function as memory size is varied.}
\label{fig:intraop}
\vspace{-10pt}
\end{figure}

\noindent \textbf{Eviction policy.}  By default, 
the scheduler eliminates the worker with the lowest-quality model replica
from the pool. If the significance filter is enabled, i.e., $v_t > 0$, 
the leaving worker $p$ stores~its local replica of the model $\mathbf{\widetilde{x}}_{t, p}$
to external storage before terminating itself.  Subsequently, each active worker $p^\prime \neq p$
downloads $\mathbf{\widetilde{x}}_{t, p}$ from external storage and averages~it with its local model, i.e., 
$\mathbf{\widetilde{x}}_{t, p^\prime} = \frac{1}{2}\left(\mathbf{\widetilde{x}}_{t, p}+ \mathbf{\widetilde{x}}_{t, p^\prime}\right)$, to 
reintegrate the non-significant updates from the leaving worker into its local model.
For $v_t = 0$, we note that the ISP model reduces to the BSP model (see \ref{sec:analysis}), and thus,
this additional one-shot synchronization is unneeded.  

\section{Implementation} 
\label{sec:impl}

We implement \mlless~by extending PyWren-IBM~\cite{pywrenibm} \textemdash~a Python-based serverless data analytics framework.
Although PyWren-IBM allows users to execute user-defined functions (UDFs) as serverless workers, it is painful slow for ML training~\cite{cirrus}.
So, to make \mlless\hphantom{} competitive with the ``serverful'' ML libraries, we reimplemented part of PyWren-IBM's runtime, the models and optimizers
(SGD, SGD~with momentum, ADAM, etc.), including sparse data structures, in Cython\footnote{\url{http://cython.org/}}, using C-style static type 
declarations that allow compilation. ML frameworks such as PyTorch rely heavily on C++ and math libraries such as Intel 
MKL\footnote{\url{https://www.nsc.liu.se/software/math-libraries/}}
to speed up computations on CPU. Thus, a pure Python implementation for \mlless~would have degraded system throughput to
a large extent.
\medskip

\noindent \textbf{Intra-level parallelism}. A final important observation to make is the lack of thread-level parallelism of IBM Cloud Functions. For the maximum
memory allocation of $2$GB, we can get the equivalent of one vCPU. This implies that we cannot exploit data parallelism 
within a worker as ML systems such as PyTorch do \textemdash~e.g., through OpenMP. To corroborate this, we ran a small
micro-benchmark. Concretely, a probabilistic matrix Factorization (PMF)~\cite{pmf} model was trained running SGD on either 
one or two threads. We measured the per-step running time of~the computations inside the workers and computed
the speedup of the two threads relative to single-threaded performance. The results are plotted in Fig.~\ref{fig:intraop}.
As can be seen in the figure, PyTorch is able to extract some parallelism within a worker,  but it is clearly not enough to exploit data parallelism.
For workers with $1536$ MiB of memory, we even found that the 
performance with $2$ threads was worse than single-threaded
performance due to a misallocation of resources.

\section{Evaluation}
\label{sec:eval}
In this section, we perform a series of experiments to answer the following main questions: 
\begin{itemize}
\item \emph{What is the individual contribution of each optimization to cost-efficiency?} For we perform a number of micro-benchmarks.
\item \emph{Is it possible to achieve better cost-efficiency with an optimized FaaS platform than a VM-based, i.e., ``serverful'' substrate  (IaaS) for distributed ML training?} For
we run several ML training jobs of different flavors, including both dense and sparse ML models. We use PyTorch~\cite{pytorch}, a specialized ``serverful'' ML library, but also 
a non-specialized, serverless data-analytics system, to determine what happens when FaaS is not specialized to model training.  
\item \emph{Is the ISP consistency model much more effective than other bounded staleness models such as SSP?} The goal is to infer what type of synchronization strategy is more appropriate  
for the indirect communication model of FaaS cloud platforms.
\end{itemize}

To conclude, we also evaluate the scalability of \mlless~on the exchange of intermediate training state, which is the main system bottleneck due to the impossibility of 
function-to-function communication. Note that the scalability of object storage for the storage of training datasets has already been assessed in other works~\cite{cirrus}.
Consistent with these works, we have observed no bottleneck for the download of mini-batches from IBM COS. 

\subsection{Methodology}

\noindent \textbf{Competing systems.}  Concretely, we compare \mlless~with the following implementations:
\begin{itemize}
  \item \textbf{Distributed PyTorch~\cite{pytorch} on CPUs}.  Due to the lack of hardware accelerators such as GPUs and TPUs~\cite{tpu} in
	IBM Cloud Functions, we run PyTorch v1.8.1 with Intel MKL enabled in a cluster of VM servers utilizing all the available cores. We use the all-reduce operator
	of Gloo~\cite{gloo}---rule of thumb for CPU training---, a MPI-like library  for cross-machine communication. 
	Mini-batches are downloaded from IBM COS.
 
   \item \textbf{PyWren-IBM~\cite{pywrenibm}}. We use PyWren-IBM as non-specialized serverless ML representative.  PyWren-IBM has been optimized
	to run on IBM Cloud Functions. Since it is a MapReduce framework,  we leverage the \texttt{map} phase to process mini-batches in parallel
	and \texttt{reduce} tasks to aggregate the local updates. All communication is done through IBM COS, including the sharing of updates, to keep
	its pure serverless, general-purpose architecture.
\end{itemize}

\noindent \textbf{Datasets.} We utilize three datasets in our evaluation. First, we use
the \textbf{Criteo} display ads dataset~\cite{criteo}, which contains $47$M samples and has $11$GB of size in total. Each sample consists of 
$13$ numerical and $26$ categorical features. Before training, we normalize the dataset. In particular, we manipulate this dataset in two forms. On one hand, we only~use the $13$ numerical features to produce
a \emph{dense} dataset. On the other hand, we hash all the categorical dimensions to a sparse vector of size $10^5$ (``hashing trick''), along with the $13$ numerical features, to produce a \emph{sparse} dataset. In this way, we can evaluate the impact of sparsity on the cost-efficiency of FaaS over IaaS as another evaluation dimension. 

Also, we use  the \textbf{MovieLens-$10$M} and \textbf{MovieLens-$20$M} datasets~\cite{movielens}. The former consists of $10$M movie reviews from $N_u = 10,681$ users 
on $N_m = 71.567$ movies. The latter bears around $20$M reviews from $N_u = 27,278$ users on $N_m = 138,493$ movies. Notice that all the datasets are (highly)-sparse to
verify \mlless~ support for sparse data. 
\medskip

\noindent \textbf{ML models.} As shown in Table~\ref{tab:eval}, 
we train different models on different datasets, i.e., Criteo for logistic regression (\textbf{LR}), and MovieLens-$10$M/$20$M for probabilistic matrix factorization (\textbf{PMF})~\cite{pmf}.  
Concretely, for \textbf{PMF}, we factorize the partially filled matrix of review ratings $\mathbf{R}$ of size $N_u \times N_m$ into  two  latent  matrices: $\mathbf{U}_{N_u \times r}$ and  $\mathbf{M}_{N_m \times r}$, such that $\mathbf{R} \approx  \mathbf{UM}$.
\medskip

\begin{table}[t]
    \centering
		\newcolumntype{P}[1]{>{\raggedright\arraybackslash}p{#1}}
		\caption{ML models, datasets, and experimental settings. $B$ means mini-batch size, and $r$ means targeted rank of PMF.} 
    \begin{tabular}{P{1cm}P{2cm}P{5cm}P{2.2cm}P{2.8cm}} \toprule
        Model & Dataset & Optimizer & \# Workers & Setting \\ \midrule
        LR & Criteo &  Adam &$12$, $24$&  $B = 6,250$ \\
        PMF & ML-$10$M & SGD + Nesterov momentum &$12$, $24$ & $B = 6,250$, $r = 20$ \\
        PMF & ML-$20$M & SGD + Nesterov momentum & $12$, $24$ & $B = 12$K, $r = 20$ \\ \bottomrule
    \end{tabular}
    \label{tab:eval}
		\vspace{-5pt}
\end{table}

\begin{table}[t]
    \centering
		\renewcommand{\arraystretch}{1.2}
    \newcolumntype{P}[1]{>{\raggedright\arraybackslash}p{#1}}
		\caption{Pricing from IBM Cloud (\texttt{us-east}, April. 2021).}
    \begin{tabular}{P{5.2cm}P{4.8cm}P{4cm}} \toprule
        Instance type & Description & Price \\ \midrule
        C1.4x4 ($4$vCPUs, $4$GB RAM) & \mlless~messaging service & $0.15$ \$/hour \\
				M1.2x16 ($2$vCPUs, $16$GB RAM) & Redis & $0.17$ \$/hour \\
				Functions ($1$vCPU, $2$GB RAM) & \mlless~worker & $3.4$x$10^{-5}$  \$/s  ($0.122$ \$/hour)\\ \midrule
				B1.4x8 ($4$vCPUs, $8$GB RAM) & Four PyTorch workers &   $0.2$ \$/hour \\ \bottomrule
    \end{tabular}
    \label{tab:pricing}
		\vspace{-10pt}
\end{table}

\noindent \textbf{Setup.}
The VM instances used for the experiments are deployed on the IBM Cloud. 
Unless otherwise noted, when running \mlless, we use two VM instances: a C1.4x4 instance ($4$vCPUs, $4$GB of RAM) to 
host the messaging service, and a single M$1.2$x$16$ instance ($2$vCPUs, $16$GB of RAM) to
deploy Redis, in addition to the chosen number of FaaS workers. To 
use as many workers for PyTorch as \mlless, the PyTorch cluster will consist of
$3$ or $6$ B1.4x8 instances ($4$vCPUs, $8$GB of RAM).  All instances have a $1$Gbps NIC.
As \mlless~workers, we use the largest-sized functions of $2$GB of memory. All VMs and
\mlless~workers are deployed on the same region (\texttt{us-east}).
\medskip

\noindent \textbf{Cost computation.} The way to account for cost  is vital to measure cost-efficiency, so we included all costs incurred
by \mlless. That is, \mlless's~cost comprised the individual cost of  each component, namely the serverless workers  plus the two VM instances:
one to host the signaling service (C1.4x4 instance), and the other to exchange the intermediate training state (M$1.2$x$16$). 
Although IBM Cloud charges hourly per VM type,
we are ``conservative'' and assume that VM cost is measured as $\$/s$. This clearly favors PyTorch,
as it equates the reservation-based model of VM instances with the ``pay-per-usage'' model of serverless
computing, whereas the price of functions per time unit is proportionally much higher than VM instances.
In practice, PyTorch would cost more. To verify this, Table~\ref{tab:pricing} reports the exact pricing of each component. As shown in this table,
a serverless worker have the same amount of provisioned resources as a PyTorch worker: $1$vCPU, $2$GB RAM. The only difference is
that while serverless workers are provisioned individually, Pytorch workers are provisioned in groups of four due to their deployment
on VM instances. Consequently, a PyTorch worker costs $\frac{0.2\$/\textrm{hour}}{4} = 0.05\$/\textrm{hour}$, which is more than
two times cheaper than a serverless worker: $0.122$ \$/hour. 

We finally observe that the use of VMs confers some extra advantage to PyTorch, as the exchange of intermediate training state across all processes (\texttt{AllReduce}) can leverage the fact that some PyTorch workers are physically located in the same machine.
\medskip

\noindent \textbf{Sanity check.} Before conducting any experiment, we first 
realized a sanity check to make sure that all the models were identical in all systems. To
this end, we fixed a random~seed, and trained all models in each system using a 
single worker. We then verified that the convergence rate at each step was exactly the same in all systems. 
This guarantees no technical advantage of one system over the other due to subtle model 
artifacts such as $\ell_1$- and $\ell_2$-regularization, etc.

\begin{figure*}[t]
\centering
\subfloat[][\textbf{\textbf{LR}, \textbf{Criteo dense}.}]{
  \includegraphics[width=0.325\textwidth]{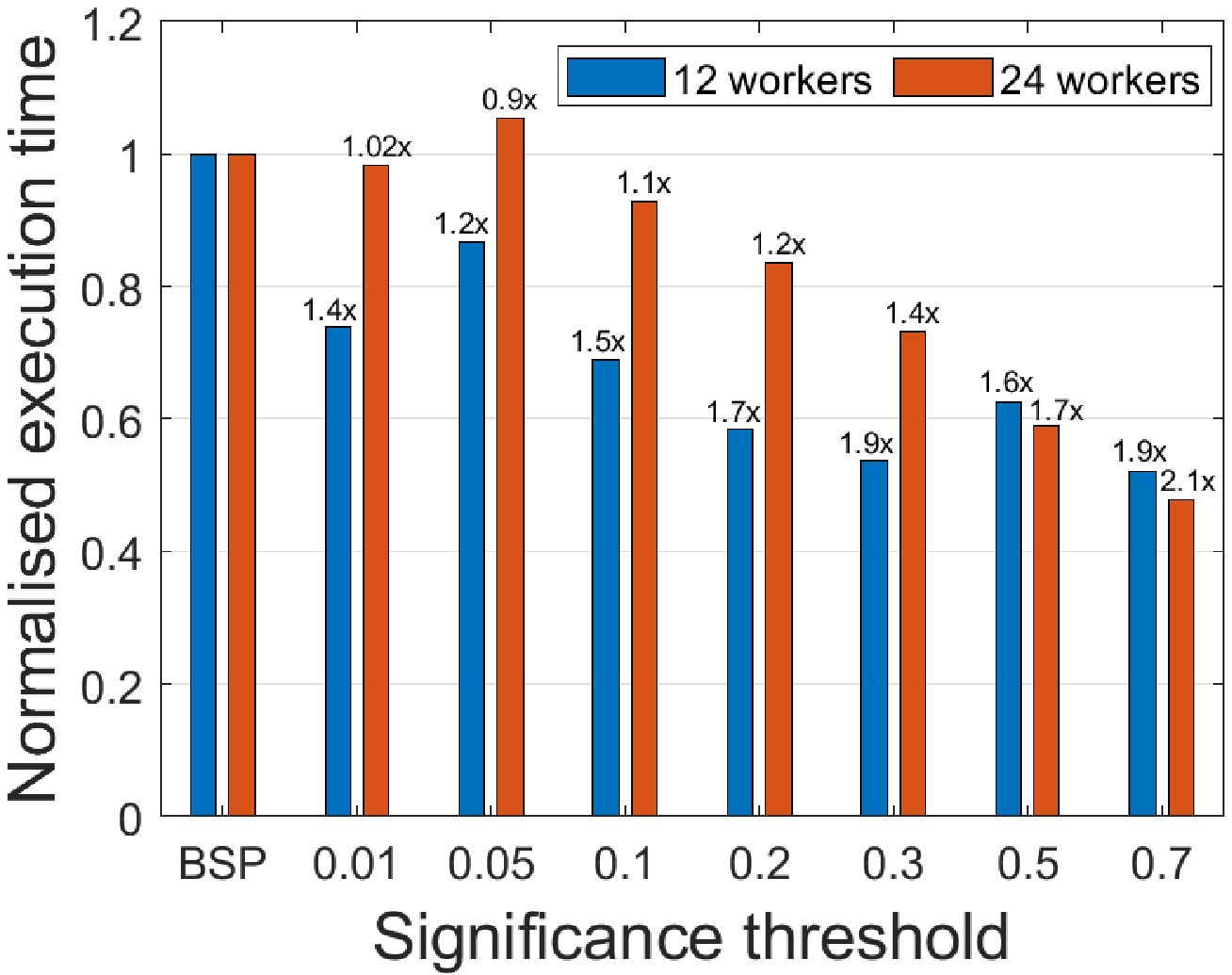}
	\label{fig:6}
}
\subfloat[][\textbf{\textbf{LR}, \textbf{Criteo sparse}.}]{
  \includegraphics[width=0.325\textwidth]{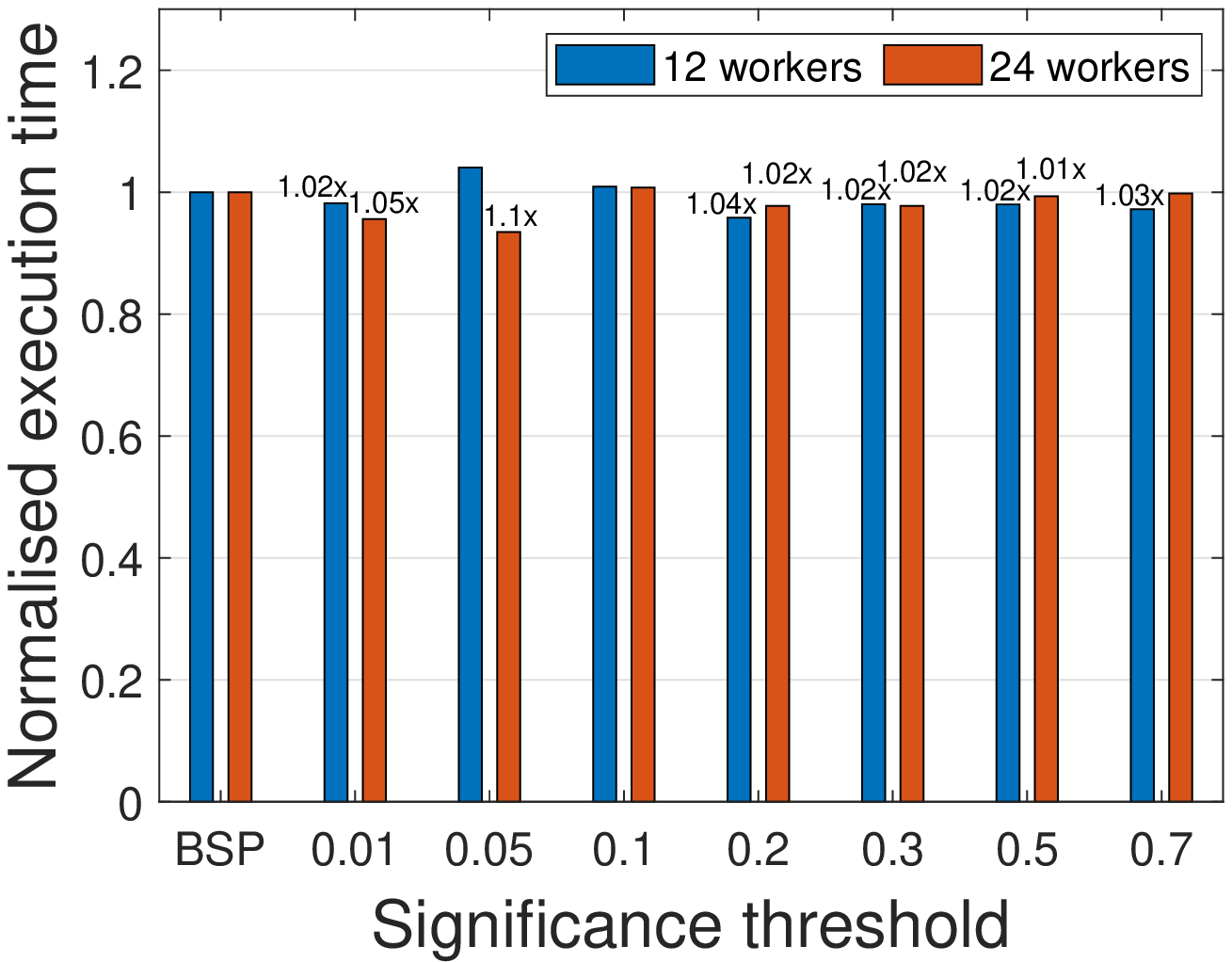}
	\label{fig:6}
}\quad
\subfloat[][\textbf{PMF}, \textbf{ML-$10$M}.]{
  \includegraphics[width=0.325\textwidth]{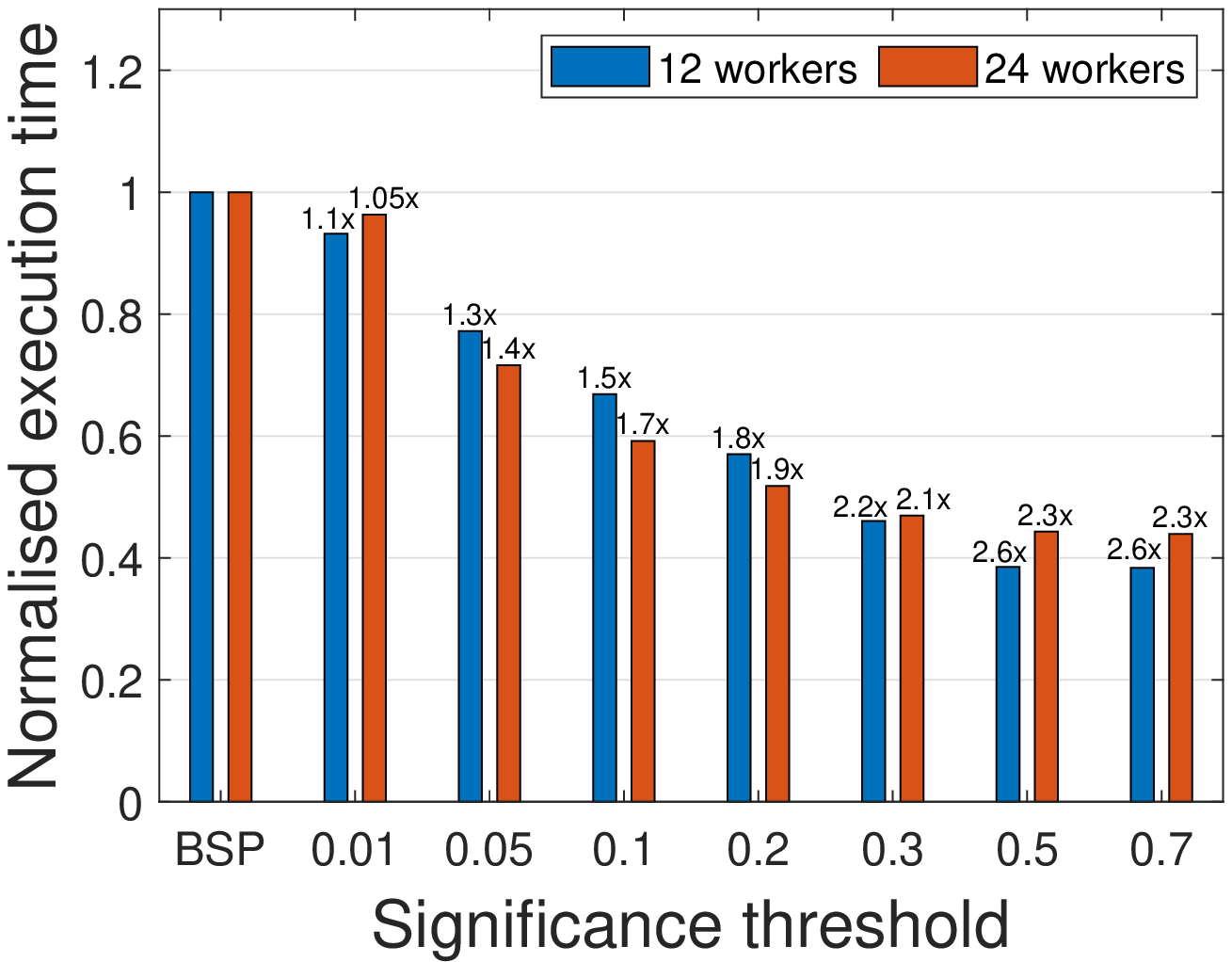}
	\label{fig:7}
}
\subfloat[][\textbf{PMF}, \textbf{ML-$20$M}.]{
  \includegraphics[width=0.325\textwidth]{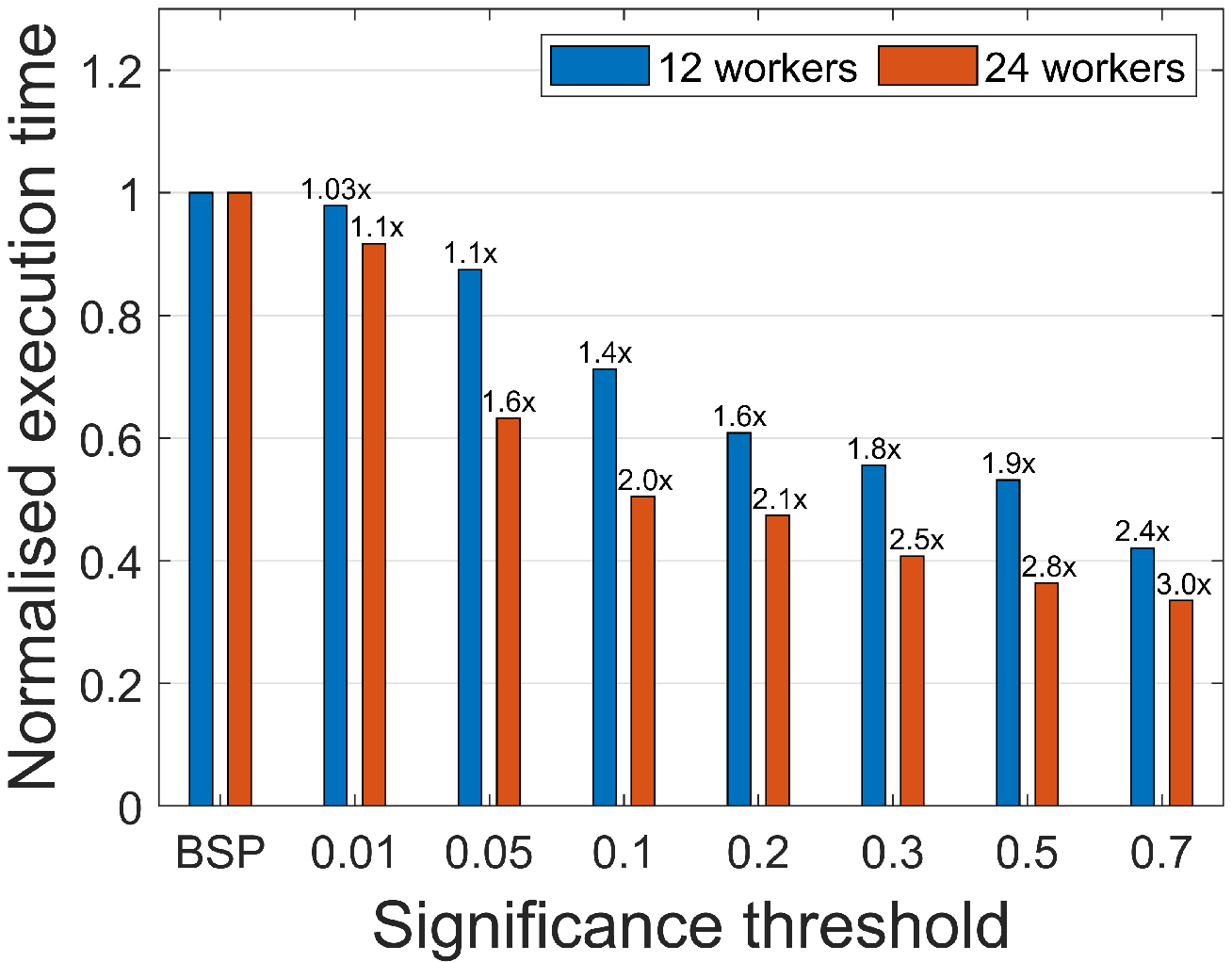}
	\label{fig:8}
}
\caption{Normalized execution time until convergence as the significance threshold $v$ increases.}
\label{fig:isp}
\vspace{-10pt}
\end{figure*}

\begin{figure}
\centering
\subfloat[][\textbf{LR}, \textbf{Criteo} (BCE=$0.58$)]{
  \includegraphics[width=0.275\textwidth]{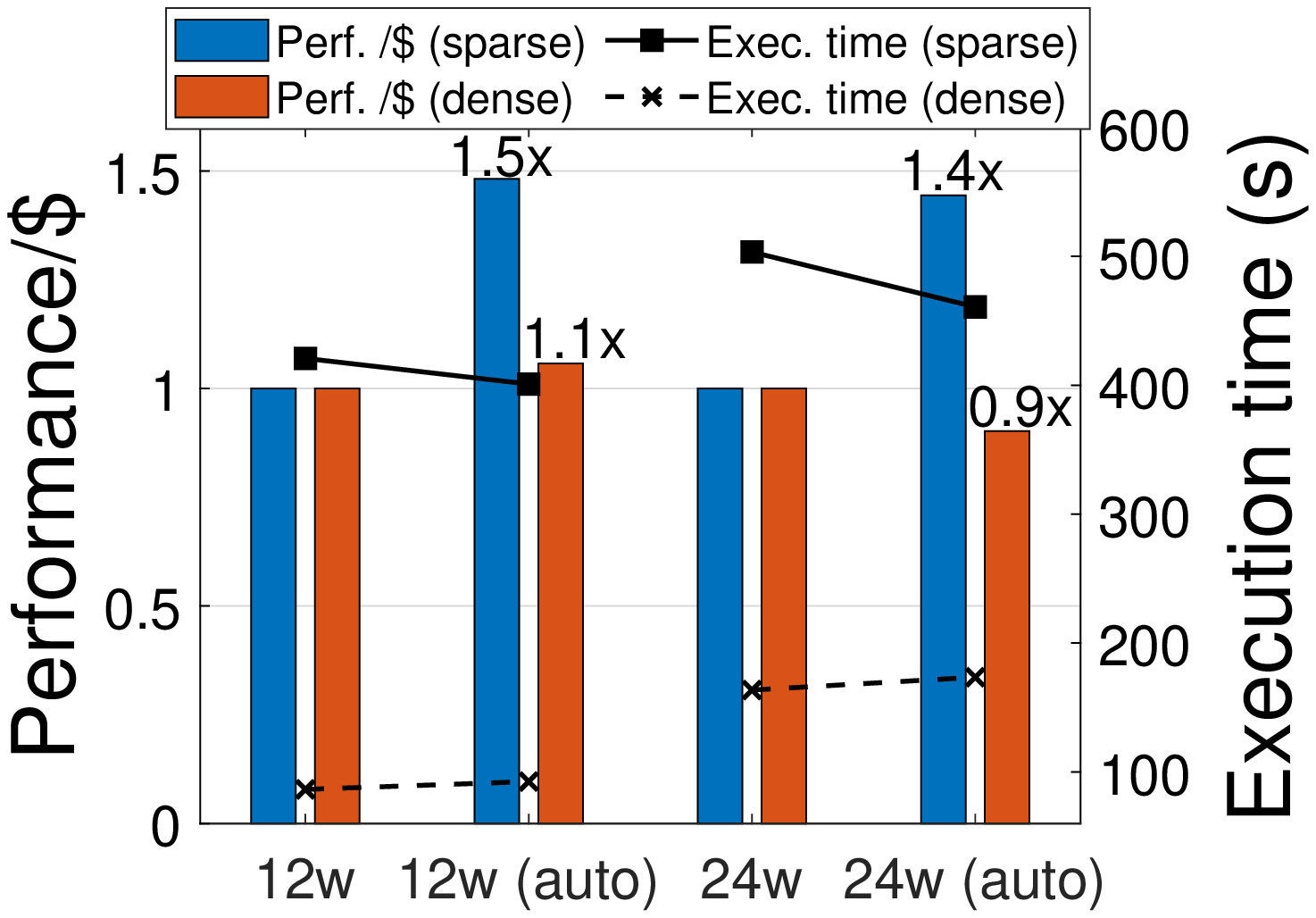}
	\label{fig:6}
} \quad
\subfloat[][\textbf{PMF}, \textbf{ML-$10$M} (RMSE=$0.72$), \\ \textbf{ML-$20$M} (RMSE=$0.77$).]{
  \includegraphics[width=0.275\textwidth]{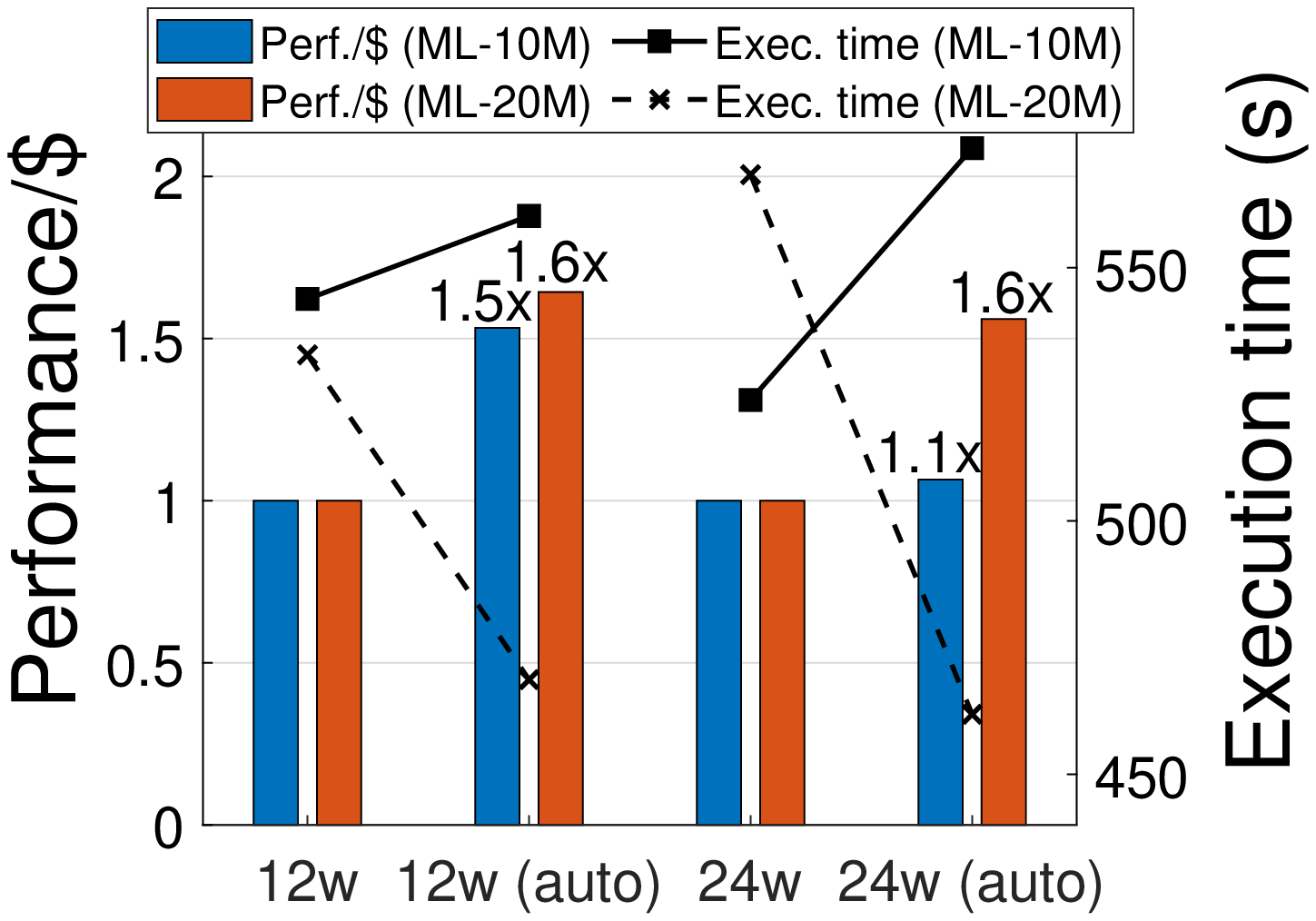}
	\label{fig:7}
}
\caption{Effect of the scale-in auto-tuner. Two metrics are used: Perf/\$\hspace{1pt}(bars; left axis); execution time\hspace{1pt}(lines; right axis).}
\label{fig:auto}
\vspace{-10pt}
\end{figure}

\subsection{Micro-benchmarks}

To better understand the individual contribution of each optimization to cost-efficiency, we run a number of micro-benchmarks.
\medskip

\subsubsection{Significance Filter}

We first evaluate the effectiveness of ISP to
improve system throughput as the significance threshold $v$ increases,~i.e., it becomes
more strict, thereby filtering out more aggressively  those updates than produce small
changes to the model. 
As a metric, we make use of 
the \textit{execution time until algorithm convergence}. For \textbf{LR}, we fix a Binary Cross Entropy (BCE) loss threshold
of $0.58$,  and stop training when the threshold is reached. For \textbf{PMF}, we set a Root Mean Squared Error (RMSE)
loss threshold of $0.82$. Because of the ``pay-as-you-go'' model of cloud functions, the key point to note here is that by \emph{decreasing the execution time, ISP cuts the cost forthwith}. 
We use the BSP synchronization model.

The results are plotted in Fig.~\ref{fig:isp}. When training PMF on both MovieLens
datasets, ISP is able to improve system throughput significantly with no side effects on
convergence. For ML-$20$M, speedup reaches $3$X. This result indicates that with
effective optimizations in communication, FaaS-based ML training can
be importantly enhanced despite the impossibility of function-to-function communication. 
The results for LR reinforce this idea and give further sense of the potential improvements 
brought by ISP. Non-surprisingly, ISP has a stronger effect on communication for Criteo dense
compared to sparse logistic regression. Actually, the difference in execution time of about $2$X between sparse and dense LR mostly lies in model sparsity. More precisely, sparse LR produces highly sparse gradients per se due to the ``hashing trick''. On one hand, \mlless~filters zeroed features, which acts as an intrinsic filter in communication and reduces the size of updates. On the other hand, the ``hashing trick"' results in 
dissimilar gradient updates, which lend themselves to little compression. Consequently, the small gains in communication end up being more significant in dense LR, despite the smaller model size of $13$ numerical features. 
\medskip

\begin{figure*}
\centering
\subfloat[][\textbf{\textbf{LR}, \textbf{Criteo dense}.}]{
  \includegraphics[width=0.325\textwidth]{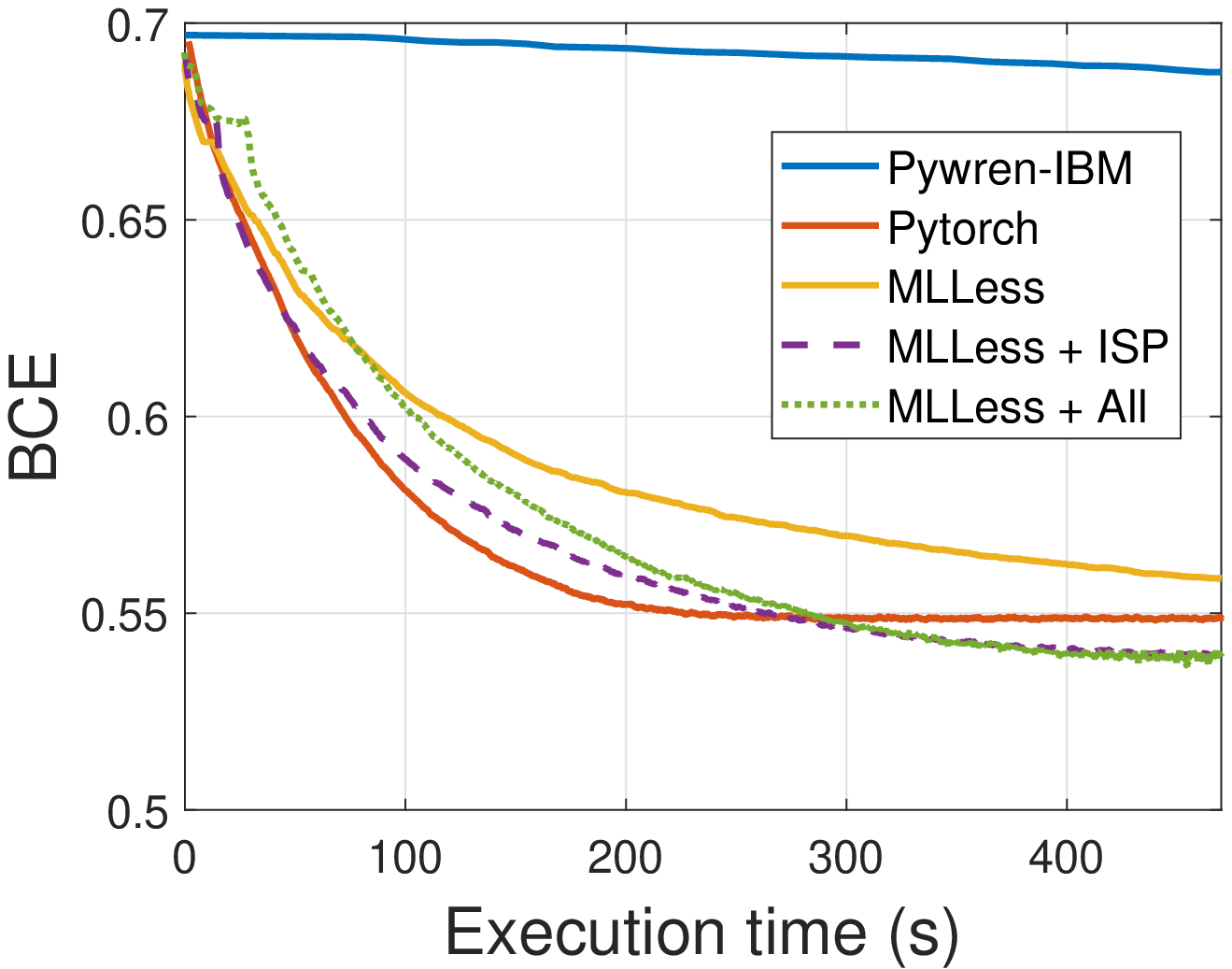}
	\label{fig:10}
}
\subfloat[][\textbf{\textbf{LR}, \textbf{Criteo sparse}.}]{
  \includegraphics[width=0.325\textwidth]{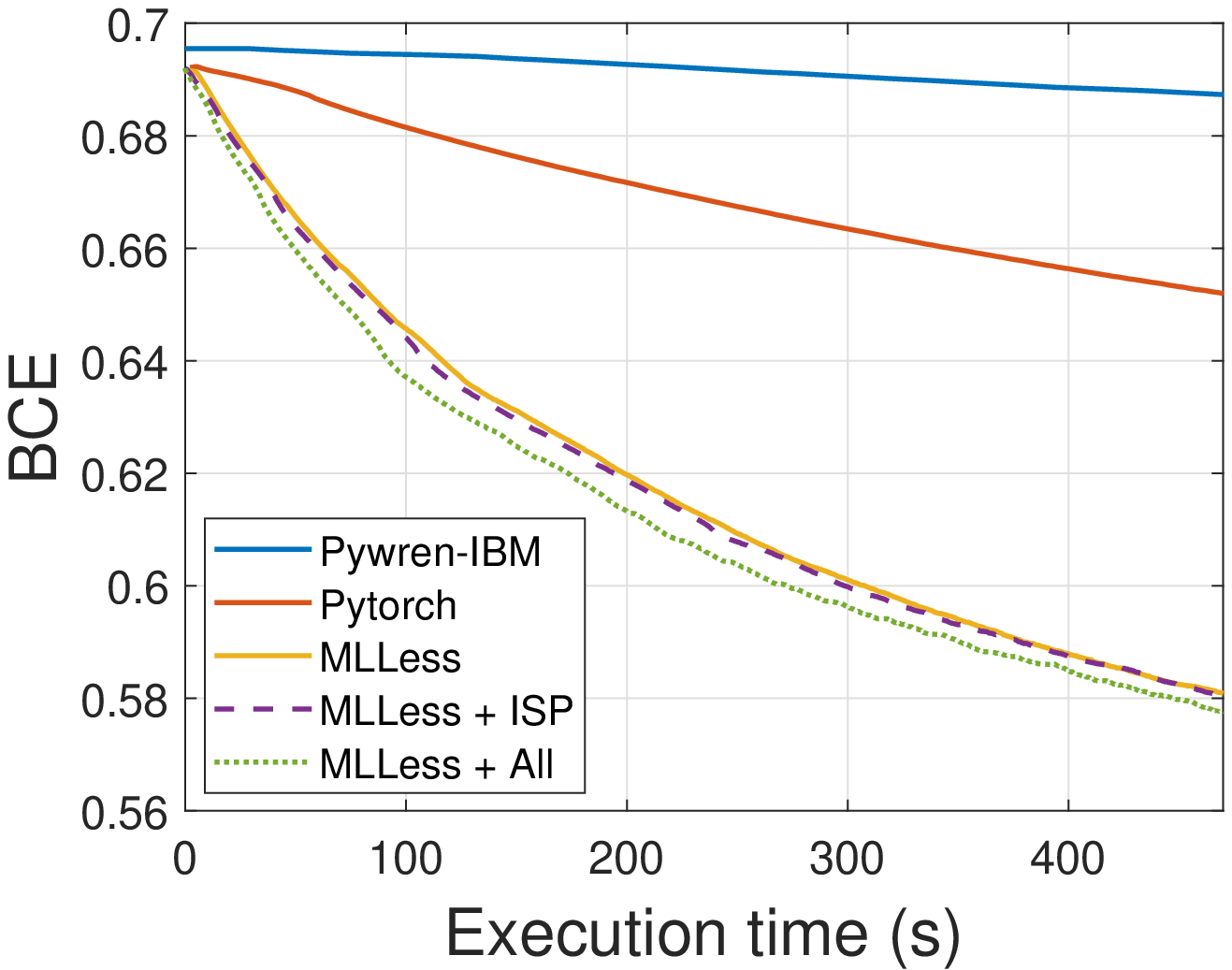}
	\label{fig:11}
}
\quad
\subfloat[][\textbf{PMF}, \textbf{ML-$10$M}.]{
  \includegraphics[width=0.325\textwidth]{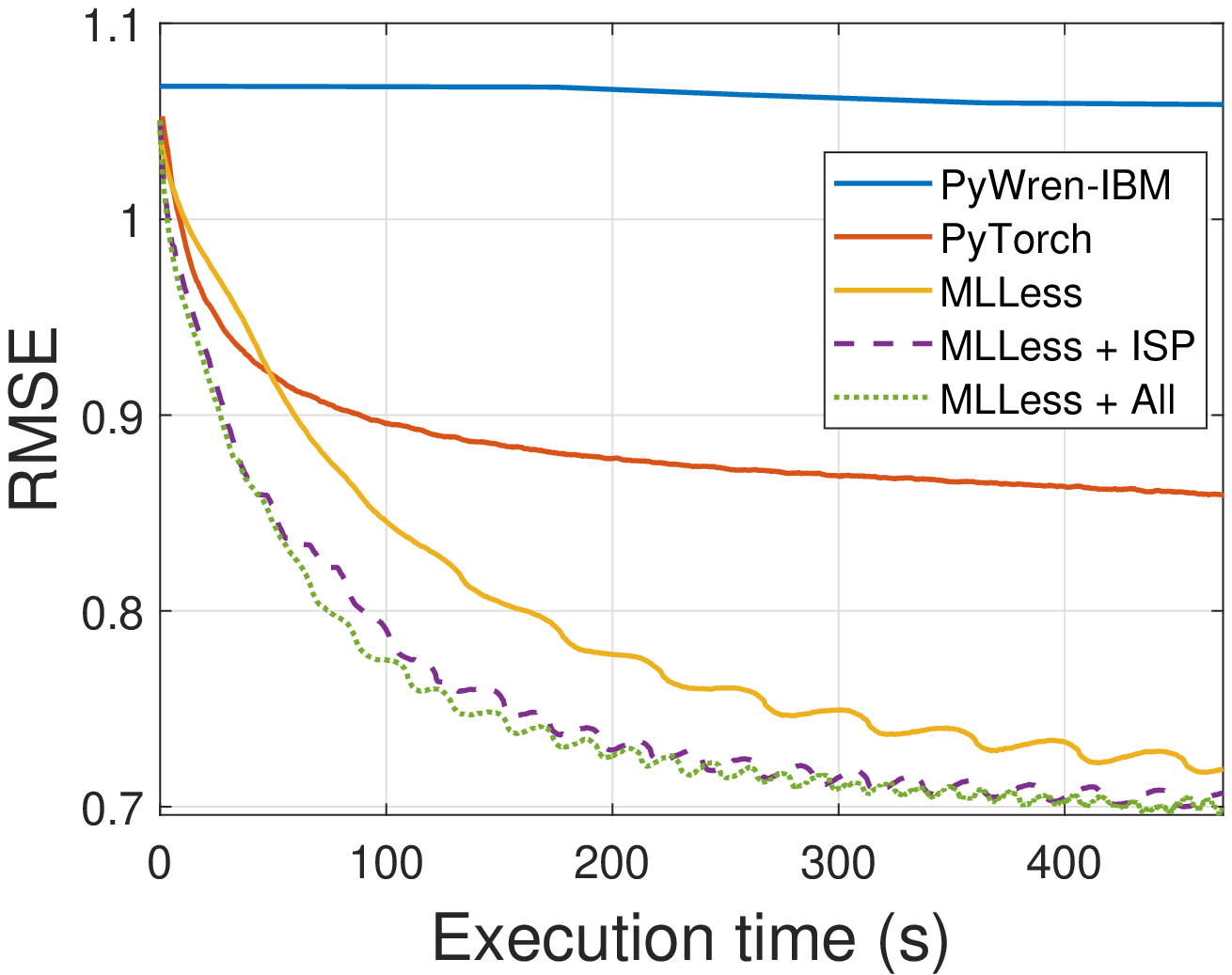}
	\label{fig:12}
}
\subfloat[][\textbf{PMF}, \textbf{ML-$20$M}.]{
  \includegraphics[width=0.325\textwidth]{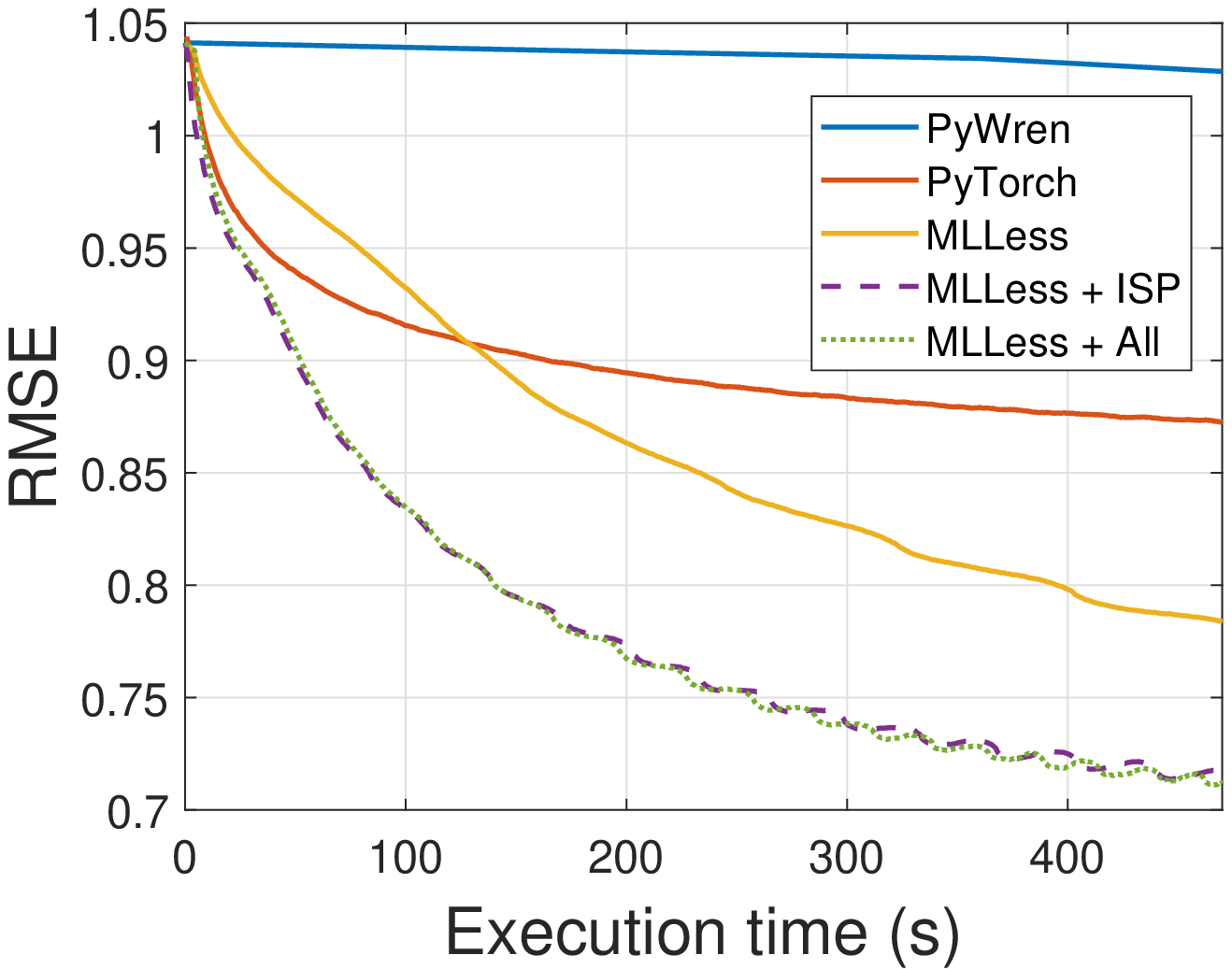}
	\label{fig:13}
}
\caption{Loss vs. time comparison between PyTorch, PyWren-IBM and \mlless~with different variants: BSP synchronization (\textbf{\mlless}), 
ISP synchronization (\textbf{\mlless~+ ISP}) and ISP synchronization + auto-tuner (\textbf{\mlless~+ All}), for $24$ workers.}
\label{fig:perf}
\vspace{-10pt}
\end{figure*}

\begin{figure*}
\centering
\subfloat[][\textbf{\textbf{LR}, \textbf{Criteo dense}.}]{
  \includegraphics[width=0.325\textwidth]{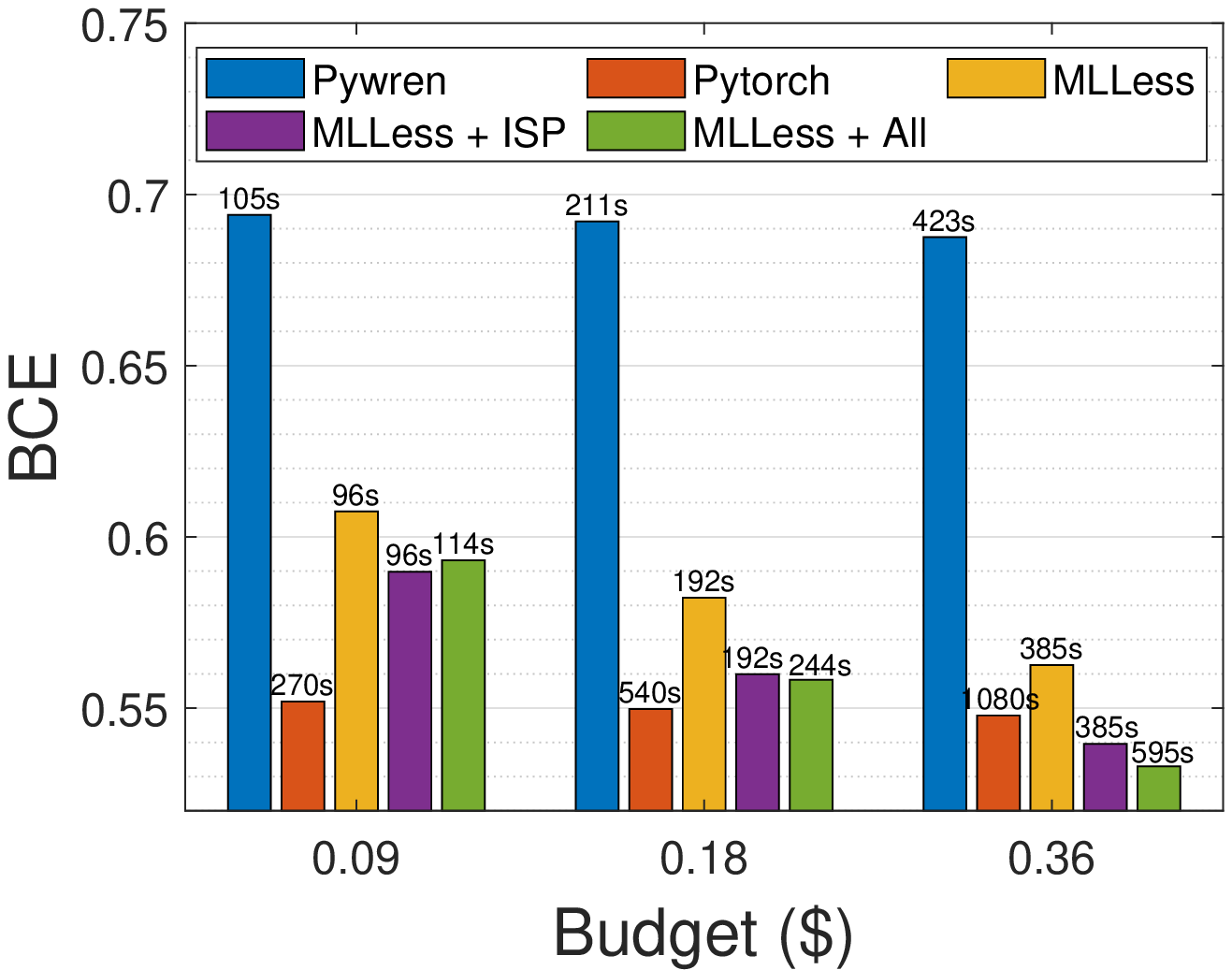}
	\label{fig:14}
}
\subfloat[][\textbf{\textbf{LR}, \textbf{Criteo sparse}.}]{
  \includegraphics[width=0.325\textwidth]{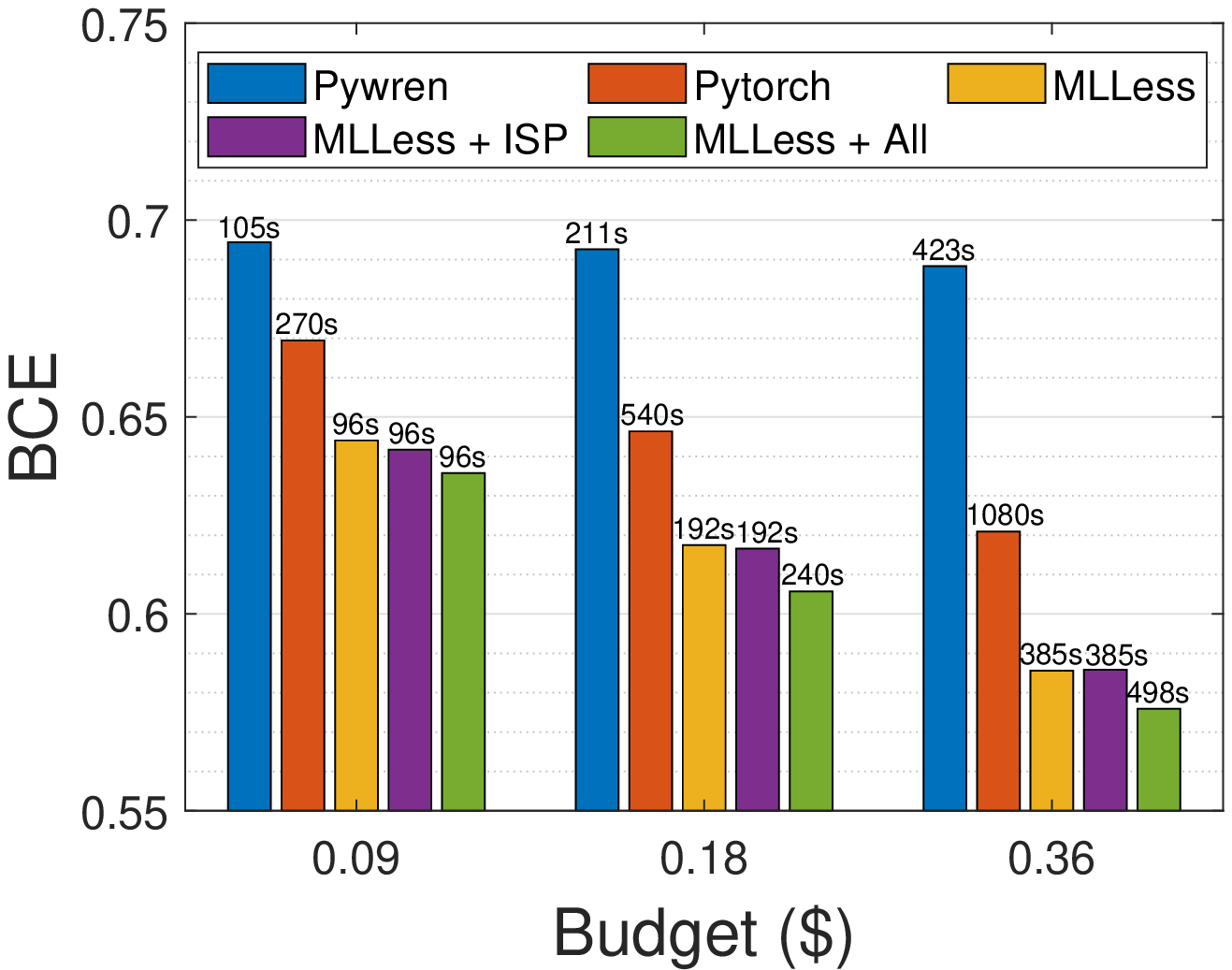}
	\label{fig:15}
}\quad
\subfloat[][\textbf{PMF}, \textbf{ML-$10$M}.]{
  \includegraphics[width=0.325\textwidth]{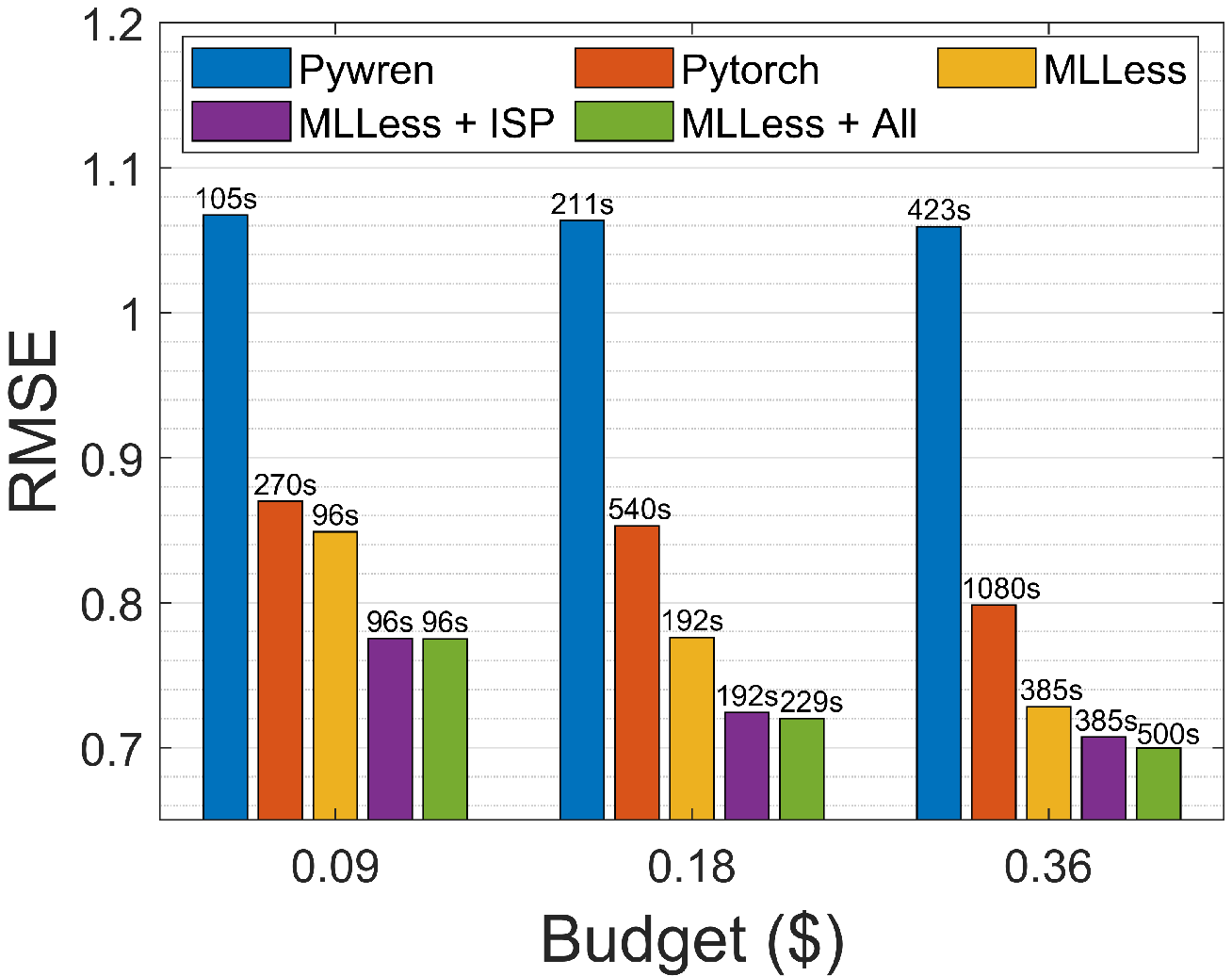}
	\label{fig:16}
}
\subfloat[][\textbf{PMF}, \textbf{ML-$20$M}.]{
  \includegraphics[width=0.325\textwidth]{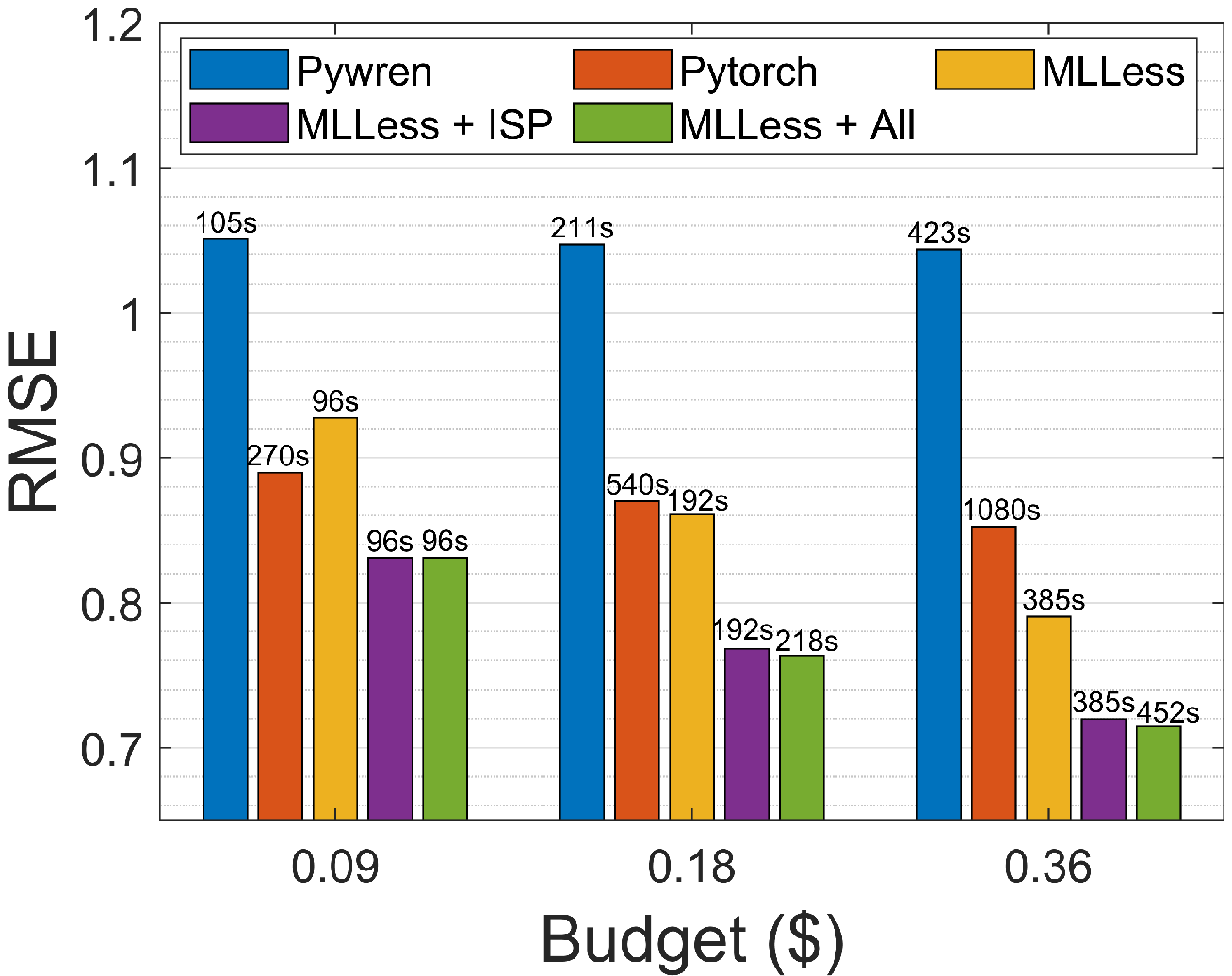}
	\label{fig:17}
}
\caption{Cost vs. loss comparison between PyTorch, PyWren-IBM and \mlless~with different variants: BSP synchronization (\textbf{\mlless}), 
ISP synchronization (\textbf{\mlless~+ ISP}) and ISP synchronization + auto-tuner (\textbf{\mlless~+ All}), for $24$ workers. The numbers above the
bars report the maximum execution time affordable with each possible budget.}
\label{fig:cost}
\vspace{-10pt}
\end{figure*}

\subsubsection{Scale-in auto-tuner.} 
\label{sec:tuner}

We assess in isolation the effect of scaling down dynamically 
the amount of workers. To draw an unbiased picture of its performance, it is insufficient to only look at the cost profile.
A bad adjustment policy could trade off convergence speed for cost, e.g., by aggressively evicting
workers from the pool. Ideally, both metrics should dwindle in parallel. To capture the effect
of the auto-tuner in a single metric, we use Perf/\$ defined as: Perf/\$$:= \frac{1}{\mathrm{Exec. time(s)}}\times\frac{1}{\mathrm{Price}(\$)}$,
so that any improvement in latency, cost, or both, caused by the auto-tuner is reflected in this composite metric. We also use~raw execution time as a
secondary metric, to detach~the \$-cost normalization effect. For Perf/\$, higher is better.
As before, we run all the ML algorithms until
convergence, defined as a threshold on the observed loss. Concrete values for thresholds
are given in the caption of Fig~\ref{fig:auto} itself. For the
scale-in auto-tuner, we set
the scheduling interval to $20$s and fix the parameter $\Delta$ at a half of  the
scheduling epoch, that is, $\Delta = 10$ sec. 

Results are illustrated in Fig~\ref{fig:auto}.
For sparse LR, the results of the auto-tuner are excellent. The auto-tuner improves the Perf/\$ between $1.4$X-$1.5$X, while
reducing the running time slightly by up to $10\%$. For dense LR, the auto-tuner~in isolation is only capable of slightly increasing the Perf/\$ for $12$ workers, leading to an improvement of
$1.1$X over the baseline. For $P = 24$ workers, the Perf/\$ worsened a little bit due to a $~9\%$ underestimation of the original convergence rate caused by an imprecise fitting of the
reference curve $L_P(t)$. We leave for future work the development of a more precise estimation method for the reference curve to prevent any degradation of  Perf/\$.

Interestingly, the fact that the execution time
increases with  more workers for the LR use case is attributable to a loss of \emph{statistical efficiency}~\cite{statistical} due to weak scaling, rather than to a
poor scalability of \mlless. To corroborate this claim, we repeated the same experiment, but now adjusting the mini-batch 
size $B$ as we varied the number of workers to keep the global batch size $B_g$ the same at all times. We got comparable results, as
listed in Table~\ref{tab:statistical}, which  shows that the converge rate was equivalent in all worker configurations for a constant $B_g$. By adapting the mini-batch size $B$, model replicas
synchronized more frequently as the number of workers grew, thus preserving statistical efficiency.
\begin{table}[h]
   \vspace{-5pt}
    \centering
		\caption{Execution time of \textbf{LR, Criteo sparse} (BCE=$0.58$) as global batch remains constant. $B$ refers to mini-batch size.} 
    \begin{tabular}{cccc} \toprule
        \# Workers & $12$ ($B = 6,250$) & $24$ ($B = 3,125$) & $48$ ($B = 1,562$) \\ \midrule
        Execution time (s) & $437.1$  & $395.3$ & $426.3$  \\ \bottomrule
    \end{tabular}
    \label{tab:statistical}
		\vspace{-5pt}
\end{table}

For PMF, the results were also nice. For all
settings, the auto-tuner improved the Perf/\$. For the ML-$20$M dataset, it even led to $1.6$X gain since it also 
delivered a significant improvement in speed. The small degradation of around $7.1\%$ in execution
time for the ML-$10$M dataset was due to an aggressive purge of the workers too much early by the auto-tuner, which can be solved
by adjusting the ``knee'' finder (see \S\ref{sec:sched}). 

As a main insight, we see that
for users who must curtail costs,\emph{ a competent exploitation of
the FaaS ``pay-as-you-go'' model  as ours can be of great help to manage their budgets}.

\subsection{Cost-Efficiency}

In this section, we explore the cost-efficiency of \mlless, while seeking to answer the nagging question of whether FaaS can outperform
a VM-based, IaaS infrastructure for distributed ML learning.

\subsubsection{Performance comparison.} 
\label{sec:performance}

To assess the  benefits of a specialized system for serverless ML training,
we compare \mlless~against PyTorch~\cite{pytorch} and PyWren-IBM~\cite{pywrenibm}. We use PyTorch as a representative of an IaaS-based~ML library.  We adopt PyWren-IBM to verify that a vanilla,~non-specialized design of
 \mlless~would have been dramatically inefficient.
                                              																						
For this experiment, we execute three variants of \mlless. The baseline version using the BSP synchronization model, and labeled `\textbf{\mlless}' in the figures.
A second variant with ISP replacing BSP, termed `\textbf{\mlless~+ ISP}', and a third one, with both optimizations all at once, labeled `\textbf{\mlless~+ All}'. 
For ISP, we set  the significance threshold $v=0.7$. For the auto-tuner, we set the scheduling epoch to $20$s with $\Delta = 10$s. For all the systems,
we only report the results for $P = 24$ workers.  The trends were similar for $12$ workers.
\medskip

\noindent\textbf{Results}. The results are shown in Fig.~\ref{fig:perf}. The first observation to~be made is that PyWren-IBM is 
very inefficient in all jobs. This is mostly due to two facts. The first is that local updates~are communicated 
across workers through slow storage only,~i.e., IBM COS. The second is the non-specialization of PyWren-IBM 
for iterative ML training. 

The second observation to be made is that \mlless~is able to converge significantly faster than PyTorch. To give a sense
of the performance gap, let us focus on the \textbf{PMF+ML-$10$M} application. To achieve a loss value of $0.9$, \mlless\phantom{} needs
$23$ seconds while PyTorch gets to this loss only after $90$ seconds. This gap increases over time and
to converge to a ``prudent'' RMSE loss of $0.738$,  PyTorch spends $2,029$ seconds.
\mlless\phantom{}, however, reaches
this loss value after $140$ seconds. This yields 
a speedup of $14.49$X. 

For the \textbf{PMF+ML-$20$M} application, we get similar results. To converge to a loss of $0.821$, PyTorch
spends $1,800$ seconds. \mlless\phantom{} achieves this loss within $115$ seconds,  $15.65$X faster than PyTorch. 
Via thorough analysis, we found that PyTorch's speed is affected by the high sparsity of the datasets as it
occurs to TensorFlow~\cite{xdl}. Unlike PyTorch, \mlless~employs Cython to directly operate on sparse data and sparse
gradients, and hence, save significant time on serializing and deserializing data. In this way, \mlless~leads to faster convergence.
Either way, the gap between plain \mlless~and the optimizations~is significant for PMF, which demonstrates that an optimized
treatment of sparsity by its own cannot realize such savings. To wit, plain \mlless~spends $334$ seconds to reduce RMSE to $0.821$, $3$X
slower than with all the optimizations present.

The \textbf{LR+Criteo dense} job produced another interesting result, which further buttresses the idea that optimizations tailored
to the FaaS environment are crucial to be competitive against IaaS-based ML training. Unlike in all the other experiments, Pytorch is able 
to outperform plain \mlless\hphantom{} in this case. However, when the \mlless\hphantom{} optimizations are enabled, 
\mlless\hphantom{} overtakes Pytorch in the middle of the execution and is able to converge to a lower BCE level.

As a final observation, it is worth to note that the auto-tuner~does not slow down convergence in any job
as shown by the `\textbf{\mlless~+ All}' curves. On the contrary, it helps to improve convergence speed in addition to
decrease cost. Also, the use of ISP consistency for large models such as ML-$20$M has been vital to ensure fast convergence
for the few initial seconds.
\medskip

\noindent\textbf{Main insight.} As a key conclusion, we find that \emph{FaaS can be more performant than IaaS}  under the same conditions (i.e., number of workers and memory per worker) for, at least,
fast-convergence models if ML training  is specialized to FaaS architectures.

\subsubsection{Cost comparison.} 

As shown above in \S\ref{sec:performance}, FaaS can outperform IaaS-based training. However, the price/time unit of a serverless worker is typically higher than IaaS-based worker. 
As given in Table~\ref{tab:pricing}, a PyTorch worker costs $\frac{0.2\$/\textrm{hour}}{4} = 0.05\$/\textrm{hour}$, which is more than two times cheaper than a serverless worker: $0.122$ \$/hour. 
Therefore, a better cost-efficiency for FaaS-based ML training is a priori more difficult, but plausible, mostly because of the possibility to dynamically adjust the number of workers, among
other abilities. 

Following the same path traced above, here we 
compare  \mlless~against PyTorch and PyWren-IBM in terms of cost. We extract the cost of
each system from the executions in the prior evaluation to ease cross comparison. 
\medskip

\noindent\textbf{General results.} As a headline observation, \mlless~is cheaper than PyTorch in all applications, but the improvement gap is not as big
as in the performance dimension. For example, when training \textbf{PMF} on\textbf{ ML-$20$M}, \mlless~spends
 \$$0.0948$ to reach a loss of $0.82$, compared to the \$$0.6$ invested by PyTorch. This leads to a $6.32$X savings on cost.
Likewise, PyTorch spends \$$0.667$ to achieve to a loss value of $0.738$ for the  \textbf{PMF+ML-$10$M} job, while 
\mlless~cuts this cost to \$$0.1348$, $4.94$X cheaper than PyTorch. 
\medskip

\noindent\textbf{Fixed-budget cost results.} While \mlless\hphantom{} saves money, for some users the ``pay-as-you-go'' model is in conflict with the way they 
manage their budgets. For instance, these may be fixed in advance. Therefore, it is interesting to examine what
would be the performance of \mlless\hphantom{} for a fixed budget. To answer this question, Fig.~\ref{fig:cost}
illustrates to what extent each system is able to converge under a fixed budget in dollars. The numbers above the
bars report the maximum execution time affordable with each possible budget. 

As can be seen in the figure, \textbf{\mlless~+ All} provides the best cost-performance trade-off in all
applications, even for the tiny budget of $9$ cents. Non-surprisingly, PyTorch is able \textit{to run longer
than the rest of systems due to the lower pricing of the rented VM instances}. For the largest budget, it 
even doubles the maximum execution time affordable~by \mlless. Per contra, \mlless\hphantom{} is significantly more efficient per time unit
and better adjusts to the cost plan. We note that the auto-tuner helps to gain some extra seconds, up to $115$
seconds, as shown by the \textbf{\mlless~+ All}-labeled bars. This is another experimental evidence of the economic
utility of our scale-in auto-tuner. 

The minimal exception to the above rule of thumb is for \textbf{LR+Criteo dense} (Fig.~\ref{fig:14}).  For this task, Pytorch is able to deliver the 
best performance for the $9$\cents- and $18$\cents-budgets, mostly because of its optimal, ring-based  \texttt{All-reduce} primitive for dense data.
Fortunately, by leveraging the combined effect of the scale-in auto-tuner and ISP, \mlless\hphantom{} manages to incrementally improve the cost-efficiency ratio, 
achieving a lower BCE value for $36$\cents.
\medskip

\noindent\textbf{Main insight.} As a main takeaway,  \emph{FaaS can be more cost-efficient than IaaS} for fast-convergent models if  FaaS-based model training  is crafted for
the serverless environment. 

\subsection{SSP vs. ISP for FaaS-based ML training}

Due to the  need of indirect communication in FaaS-based ML training, another interesting question is to ascertain whether ISP is better suited for serverless model training
than other popular yet loose consistency models such as SSP \cite{ssp}. To this goal, we integrated SSP into \mlless, and compared it with ISP, as well as with
the baseline BSP-based version. To carry out this comparison, we experimented with PMF on the \textbf{ML-20M} dataset for an increasing number of workers $P$. To not compromise
statistical efficiency due to weak scaling (see \S\ref{sec:tuner} and Table~\ref{tab:statistical} for further details),  we fixed the global batch size $B_g$ and adjusted the mini-batch size $B$
accordingly. Concretely, we set $B=12$K for $12$ workers, $B=6$K for $24$ workers and $B=3$K for $48$ workers. In this way, we made sure that the effect of stale updates came out neatly for each synchronization model. For SSP, we set a slack of $s  = 3$ iterations. 
\medskip

\begin{figure}
\centering
 \includegraphics[width=0.33\textwidth]{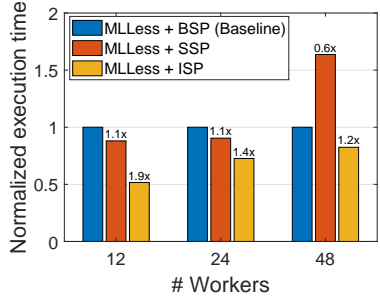}
\caption{Comparison of \textbf{SSP} against \textbf{ISP} synchronization model for serverless ML training.}
\label{fig:ssp}
\vspace{-10pt}
\end{figure}

\noindent\textbf{Results.} The results are depicted in Fig.~\ref{fig:ssp}. As expected, SSP shows a better average speedup of $1.1$X over the default BSP implementation
for $12$ and $24$ workers. For $48$ workers, SSP, however, performs worse than the synchronous~BSP model due to the lack of intra-function parallelism. More technically, 
SSP is agnostic to the computation capacity of workers, but merely ensures that the number of iterations between the fastest and the slowest workers does not exceed the \emph{staleness bound} $s$.
In a distributed setting such as that of  \mlless, where each worker is responsible to aggregate the updates from the rest, i.e., there are no global parameters, the lack of intra-function parallelism
means there is no way for the slowest workers to hide the latency of downloading and applying the missing updates. A solution to this problem would be to use a serverless backend such as 
Crucial~\cite{crucial} to perform the storage-side aggregation of gradients.

The more relevant finding of this experiment is, however, that ISP outperforms SSP in all cases, yielding a speedup of $1.9$X and $1.4$X for $12$ and $24$ workers, respectively. The reason why
ISP is way better than SSP is that ISP permits any worker to delay the synchronization of a parameter \textit{indefinitely} as long as its aggregated update is non-significant. Under SSP, however, update synchronization is delayed up to most $s$ iterations for the fastest workers, but sooner or later, all the committed updates from the workers are added to the model replicas, thus not reducing
the communication overhead at all. Put another way, the loose synchronization property of SSP is not enough to outweigh the reduction in network traffic achieved by ISP, being the latter more
effective to yield faster convergence for FaaS-based ML training due to the need of indirect communication. 
\medskip

\noindent\textbf{Main insight.} For FaaS-based model training, where the exchange of intermediate updates is the main limiting factor,  \emph{parameter staleness is not of practical utility unless it saves
network traffic}.

\subsection{Scalability}

Finally, scalability is a critical property of any ML training system, and \mlless\hphantom{} is not the exception. 
Cloud object storage as a means to hold mini-batches has proven to provide good scalability~\cite{cirrus}, and we empirically found
that the signaling channel was able to support thousands of messages per second, enough to scale to hundreds of concurrent workers. 
Certainly, among all the \mlless\hphantom{} components, we observed that the indirect communication channel built to exchange updates is the one subjected to a major strain.  
Fortunately, this channel can be easily ``scaled out''~by~adding more Redis instances and ``sharding'' intermediate updates over the pool of servers using the worker IDs. To verify~this
claim, we ran multiple training jobs with the \textbf{ML-20M} dataset for an increasing number of workers. For each worker size, we trained the model with $1$ and $2$ Redis instances,
stopping at a RMSE value of $0.77$ in all settings. To preserve statistical efficiency, i.e., a similar convergence progress per second, we adjusted the batch size as in the prior test. 
\medskip
 
\begin{figure*}
\centering
\subfloat[][Effect of Redis instances.]{
  \includegraphics[width=0.325\textwidth]{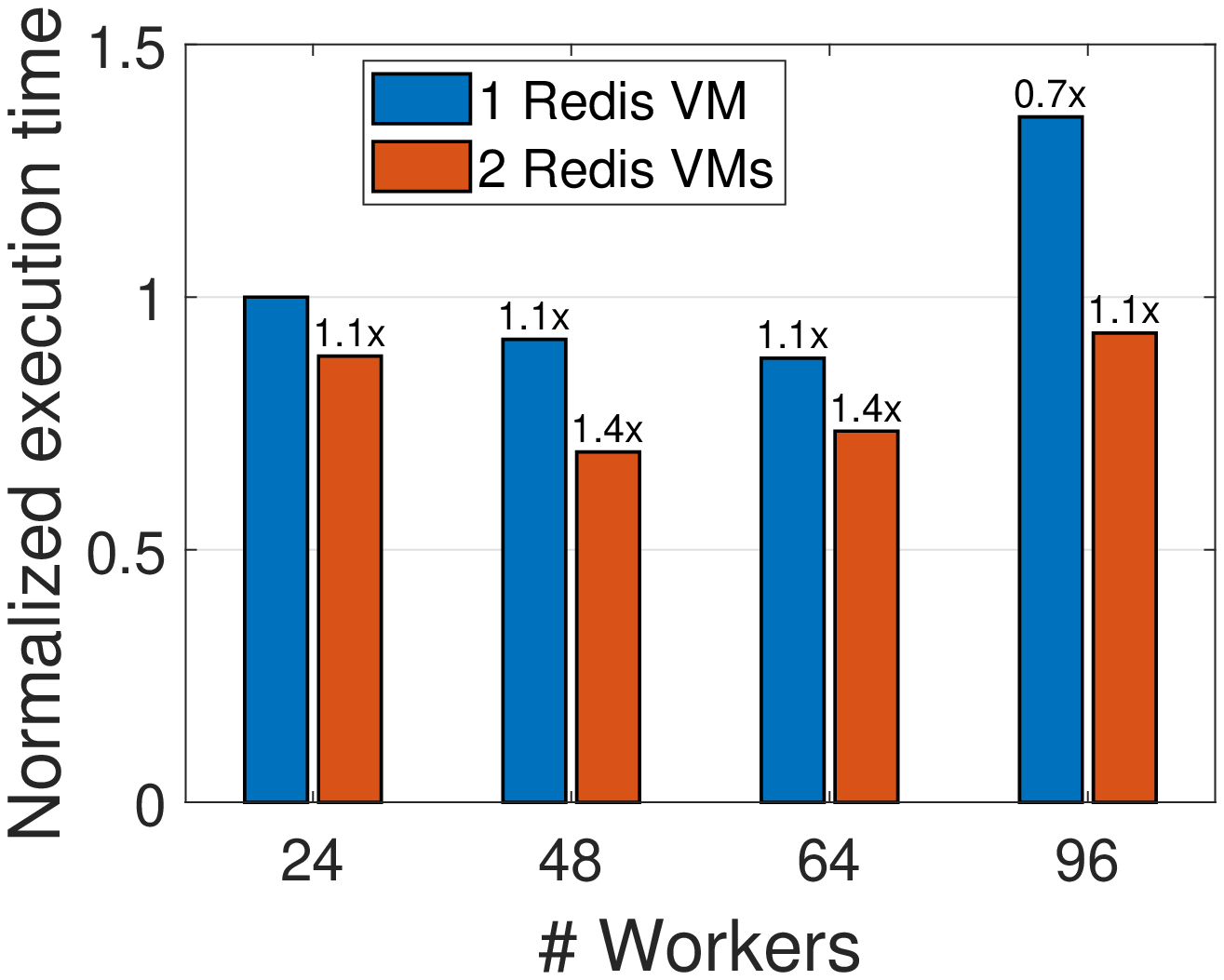}
	\label{fig:20a}
}
\subfloat[][Statistical efficiency vs. \# of workers.]{
  \includegraphics[width=0.325\textwidth]{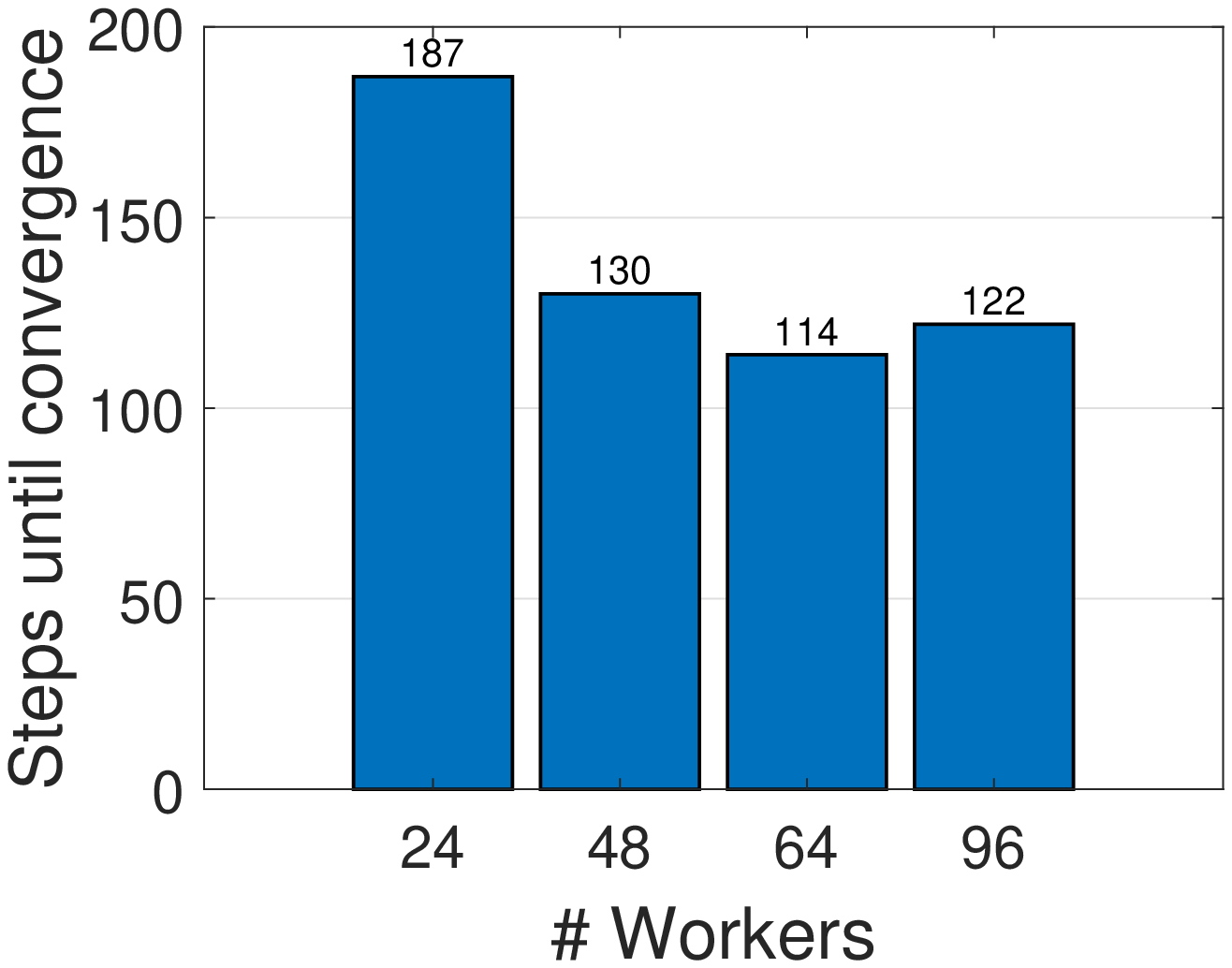}
	\label{fig:20b}
}
\caption{Scalability of \mlless\hphantom{} on the \textbf{ML-20M} dataset.}
\vspace{-10pt}
\end{figure*}

\noindent\textbf{Results.} The scalability results are illustrated in Fig.~\ref{fig:20a}. For ease of comparison, we normalized the execution times, choosing the configuration of $24$ workers with $1$ Redis server as the baseline. Non-surprisingly, doubling the number of Redis instances has almost no effect for a small number of workers. Nonetheless, its effect becomes more apparent as the number of workers increases and a single server cannot keep up with the high rate of updates. With two Redis servers, \mlless\hphantom{} is able to deliver a speedup of $1.4$X  for $64$ workers, despite a super-linear increase in the relation of the number of workers to the Redis servers with respect to the baseline setup, i.e., $\frac{64~\mathrm{\#workers}}{2~\mathrm{\#Redis~servers}}/\frac{24~\mathrm{\#workers}}{1~\mathrm{\#Redis~servers}} = 1.33$. Likewise for $96$ workers, the execution time is a $10\%$ better with $2$ Redis servers,  despite featuring $4$X more number of workers than the baseline setup. 
This confirms that the addition of more servers enables \mlless\hphantom{} to scale to a larger number of workers. 

It is worth to mention here that as in IaaS-based ML training,  the scaling of the training process in FaaS platforms~is equally challenging. Simply put, users expect the training time to go down 
with the number of workers. However,~even if the sharing of updates is not a bottleneck with the addition of more Redis instances, the relation between the mini-batch size and the number of workers may hinder \textit{statistical efficiency}. This is reflected in Fig.~\ref{fig:20b}, where the number of training steps until reaching the threshold is shown. We use the number of training steps,  instead of time,~as~a~metric, because the total number of steps to convergence is independent of the size of the Redis cluster. It is a quality measure that depends on the global batch size and the number of workers, which captures very well how frequently~the~workers synchronize in relation to the processed training data.

As seen in this figure, the number of iterations to convergence decreases until $96$
workers, the point beyond which adding more workers starts hurting convergence. At this point, adding more Redis servers can help reduce I/O time, so as the training time. But  eventually, the increase in the number of iterations caused by the usage of more workers will limit scalability. For this reason, users try to compensate this reduction in statistical efficiency
by either increasing the learning rate \cite{lr}, or adjusting the batch size adaptively~\cite{2018dont}.
\medskip

\noindent\textbf{Main insight.}  We conclude that \mlless\hphantom{} is scalable, and that can be easily ``scaled out'' through in-memory storage sharding. As in traditional VM-based systems, the ultimate scaling of the training process depends on the mini-batch size, the synchronization model, etc., irrespective of whether storage sharding can eliminate the bottlenecks.

\section{Related Work}
\label{sec:rw}

\noindent\noindent{\textbf{Serverless Data Processing.}} A large bulk of
previous works have proposed high-level frameworks for running large-scale analytics on serverless functions. For example, PyWren~\cite{pywren}, IBM-PyWren\cite{pywrenibm} and Lithops~\cite{multicloud, lithops} are map-reduce frameworks running over FaaS executors that take advantage of object storage to store intermediate data. Lithops~\cite{multicloud, lithops} is multi-cloud and also implements the 
native \texttt{multiprocessing} module available in Python to enable the transparent execution of multiprocessing applications over FaaS platforms. Further,
\texttt{gg}~\cite{gg} is a library that uses AWS Lambda for CPU-bound intensive jobs such as video-encoding. Numpywren~\cite{numpywren} is an elastic linear 
algebra library on top of a pure serverless architecture. Starling \cite{starling} proposes a serverless query execution engine.
Serverless ML systems, including \mlless, build upon the lessons learned from these works to increase their performance and cost-efficiency.

~Another important work is Crucial~\cite{crucial, crucial22}. Crucial is a framework for building stateful FaaS-based multi-threaded applications, and as such,
it includes fine-grained synchronization primitives such as semaphores and barriers.  As part of its evaluation, Crucial was compared
to Spark using two classical ML algorithms: K-means and logistic regression, showing an on-par performance with Spark.  Although Crucial is not cost-efficient per se, 
we believe that it would be~a good option to implement a parameter server-like interface for server ML training. 
 \smallskip

\noindent\textbf{Serverless ML.}  A number of works have been devoted to leveraging 
FaaS platforms for building ML systems. Since ML model inference is a representative
use case of serverless computing~\cite{stratum, vatche18}, recent research efforts
have been directed towards model training~\cite{cirrus, siren, lambdaml}. 
All these works use AWS Lambda, which confers them some advantage over \mlless,
and make direct comparison problematic. First, AWS Lambda enables multi-threaded
parallelism~\cite{lambada, cirrus}, while  IBM Cloud Functions are limited to $1$vCPU at most. Also,  AWS Lambda
workers can access $10$GB of local RAM, which allows them to hold larger data partitions and mini-batches
compared with \mlless\hphantom{} that is restricted to $2$GB of memory.
Despite this, \mlless\hphantom{} outweighs these limitations and manages
to deliver speedups superior to $15$X and with $6.3$X lower cost than PyTorch, and very importantly,
\textit{excluding the start-up time,} which is longer in PyTorch (e.g., a cluster of $6$ VMs takes
$>1$ min. to boot up). 

To put in a nutshell, Cirrus~\cite{cirrus} is a serverless ML system that implements a parameter server (PS)~\cite{parameter}
on top of VMs, where all FaaS workers communicate with this centralized PS layer. Such a hybrid design
has its merit, mainly because the ability of PS servers of doing computation delivers $200\%$ communication
savings compared with indirect communication via external storage. According to~\cite{berkeley}, 
Cirrus is $3$X-$5$X faster than VMs, but up to $7$X more costly. Compared to \mlless, Cirrus is thus not cost-efficient, mostly because 
it does not exploit well the definitory  properties of the FaaS model such as ``pay-per-usage'', which allows to save money through
the fine-grained, dynamic allocation of serverless workers.

\textsc{Siren}~\cite{siren} presents an asynchronous ML framework,  where each worker 
runs independently, i.e., it reads a (stale) model from remote storage (e.g., AWS S$3$), updates it
with a mini-batch of local data, writing the  new model back to storage. Its major strength is
withal its scheduler built upon reinforcement learning (RL) that
adjusts the number of workers dynamically, subject to a certain budget. Compared to
\mlless, its scheduler is more coarse-grained as it adjusts the number of workers once
per epoch, and achieves a lower cost-efficient ratio. Concretely, \textsc{Siren} reduces 
job execution time by up to $44.3\%$ at the same cost than EC$2$ clusters.

Finally, \cite{lambdaml} proposes LambdaML, a FaaS-based training system to 
determine the cases where FaaS holds a sway over IaaS. Interestingly, our results mirrors their observation that
FaaS is more cost-efficient for models that quickly converge. Unlike LambdaML, however, we reach the same conclusion
by getting out of the equation start-up times. That is, if the start-up time is excluded, LamdaML is slower
than PyTorch, while \mlless~outperforms PyTorch, yet being cheaper. In this sense, we believe that \mlless\hphantom{} opens
the door to the adoption of serverless ML training as a truly cost-efficient option in the cloud.

\section{Conclusion and Future Work}
\label{sec:end}

We have examined the question of whether serverless ML
training can be more cost-effective than traditional  IaaS-based 
computing. To answer this question, we have developed  \mlless, a prototype system 
of FaaS-based ML model training built on top of IBM Cloud Functions, and empowered it with two
new optimizations: one aimed to reduce communication bandwidth, the other intended to
exploit~the essential qualities of the FaaS model to jointly decrease cost and execution time. 
Our results demonstrate that \mlless\hphantom{} is more cost-efficient than serverful ML libraries at a lower cost
for ML models with fast convergence.  We also validate the scalability of \mlless, and the benefits of loose
synchronization models that allow to smoothly trade off communication bandwidth and convergence
time.

The cost-effectiveness of serverless ML training suggests a variety of potential future work. An interesting avenue
of research would be to investigate the advantage of supporting ML-specific logic on the server side through
serverless data stores (e.g., Crucial~\cite{crucial, crucial22}). Another topic would be to adapt \mlless\hphantom{} to Federated Learning
(FL) environments. A commonplace practice in FL is to select a random subset of the available clients in each training step, which results in 
many clients staying idle for a long time. By extending \mlless\hphantom{} to run the clients as functions on edge devices only when needed,
it would be possible to improve cost-efficiency. A final research topic would be to examine the potential effects of lossy gradient compression techniques
such as gradient quantization on serverless ML training.
\medskip

\section*{Acknowledgments}
This work has been partially supported by EU (No. $825184$) and 
Spanish Government (No. PID2019-106774RB-C22).  Marc S\'{a}nchez-Artigas is a Serra H\'{u}nter Fellow. 
Pablo Gimeno is a Mart\'{i} Franqu\`{e}s Research Grant Fellow - Banco Santander Edition.

\bibliography{mlless}

\appendix
\section{Convergence Analysis}
\label{sec:analysis}

In this section, we show that the SGD algorithm for convex objectives is ensured to converge
under our consistency model. Recall that a step $t$ of the SGD algorithm is defined as: 
\[
\mathbf{x}_t=\mathbf{x}_{t-1} -  \eta_t\nabla f_t(\mathbf{x}_t) =\mathbf{x}_{t-1} - \eta_t \mathbf{g}_t= \mathbf{x}_{t-1} + \mathbf{u}_t,
\]

where $n_t$ is the step size and $\mathbf{u}_t := -\eta_t \mathbf{g}_t$ is the update of step $t$.

As described in Section~\ref{sec:sf},  the accumulated updates are  broadcast to the rest of workers only in the case 
that~they are significant. This means that significant updates are always seen by all workers. Nevertheless, insignificant
updates remain local to the workers, so different workers will ``see'' different, noisy versions of the true state $\mathbf{x}_t$.

To formally capture the difference between the ``true'' state $\mathbf{x}_t$ and the noisy views, 
let  us  define  an  order  of  the  updates  up  to  step $t$. Suppose that the algorithm is distributed across $P$ 
workers, and the logical clocks that mark progress start at $0$. Then, 
\[
\mathbf{u}_t := \mathbf{u}_{p, c} := \mathbf{u}_{t \bmod P , \left\lfloor \frac{t}{P} \right\rfloor},
\]
defines a mapping between step $t$ and $[0, P-1] \times \mathbb{N}_0$, which loops through
 clocks  ($c=\left\lfloor \frac{t}{P} \right\rfloor$), and  for  each  clock $c$  loops  through 
workers ($p = t \bmod P$).

We now define a reference sequence of states that a single-worker serial execution would follow 
if the updates were to be seen under the above ordering: $\mathbf{x}_t = \mathbf{x}_0 + \sum^{t^\prime}_{t=0} \mathbf{u}_t$.
Let $\mathcal S_{p, c}$ denote the set of significant updates propagated by worker $p$ up through clock $c$. 
Similarly, let $\mathcal I_{p, c}$ denote the set of the insignificant updates up through clock $c$ not broadcast from worker $p$.
Clearly, $\mathcal S_{p, c}$ and $\mathcal I_{p, c}$ are disjoint, and their union includes all the updates
accumulated by $p$ until exactly clock $c$.

Using the above notation, we define the noisy view $\widetilde{\mathbf{x}}_t$ as:
\begin{equation}
 \label{eq:1}
\widetilde{\mathbf{x}}_{p, c} := \mathbf{x}_0 +  \sum^{c}_{c^\prime=0}\mathbf{u}_{p, c} +   \sum_{p^\prime \neq p} \sum_{i\in \mathcal{S}_{ p^\prime, c}} \mathbf{u}_i ,
\end{equation}

where $\mathbf{x}_0$ are the initial parameters, the second term refers to the local updates applied by worker $p$, and the last term aggregates all the significant
updates shared by the rest of workers other than $p$. 

Finally, by using Eq.~(\ref{eq:1}), the difference between the ``true'' view $\mathbf{x}_t$ and the noisy view 
$\widetilde{\mathbf{x}}_t$ becomes:
\begin{equation}
\label{eq:2}
\widetilde{\mathbf{x}}_t - \mathbf{x}_t = \widetilde{\mathbf{x}}_{p, c} - \mathbf{x}_t = \widetilde{\mathbf{x}}_{t \bmod P , \left\lfloor \frac{t}{P} \right\rfloor} - \mathbf{x}_t = -\sum_{p^\prime \neq p} \sum_{i\in \mathcal{I}_{ p^\prime, c}} \mathbf{u}_i
\end{equation}
Equipped with Eq.~(\ref{eq:2}), we are now ready to start the proof~of Theorem~\ref{thm:1}.

Similarly to~\cite{ssp}, the Insignificance-bounded Synchronous Parallel (ISP) generalizes the BSP model:
\begin{cor}
For zero significance threshold $v = 0$, ISP reduces to BSP.
\end{cor}
\begin{proof}
Observe that $v = 0$ implies that the set $\mathcal I_{p, c} = \emptyset$ at all clocks, so that 
 $\sum_{p^\prime \neq p} \sum_{i\in \mathcal{S}_{ p^\prime, c}} \mathbf{u}_i = \sum^{c}_{c^\prime=0} \sum_{p^\prime < p} \mathbf{u}_{p^\prime, c^\prime}$.
Therefore, $\widetilde{\mathbf{x}}_{p, c}$ exactly consists of all updates until the current clock.
\qed\end{proof}

\begin{customthm}{1}
Suppose we want to find the minimizer $\mathbf{x}^*$ of a convex function $f(\mathbf{x}) =\sum^T_{t=1} f_t(\mathbf{x})$ (components $f_t$ are also convex) via
SGD on one component $\nabla f_t$ at a time. Also, the algorithm is replicated across $P$ workers with synchronization at every step $t$. Let
$\mathbf{u}_t := - \eta_t \nabla f_t(\widetilde{\mathbf{x}}_t)$, where the step size $\eta_t$ decreases as $\eta_t=\frac{\eta}{\sqrt{t}}$. As per-parameter significance
filter, we use $\left| \frac{\delta_{i, t}}{\widetilde{x}_{i, t}}  \right| > v_t$,  where $\widetilde{x}_{i, t}$ is the $i^\mathrm{th}$ parameter of the noisy
state $\widetilde{\mathbf{x}}_t := (\widetilde{x}_{0, t}, \widetilde{x}_{1, t}, \dots, \widetilde{x}_{n, t})$ at step $t$, $\delta_{i, t} := \sum^t_{t^\prime = t_{p_i}} u_{i, t^\prime}$ denotes the
accumulated update for the $i^\mathrm{th}$ parameter since the last propagation time $t_{p_i}$, and $v_t$ is the significance threshold that decreases as $v_t=\frac{v}{\sqrt{t}}$.
Then, under suitable conditions: $f_t$ are $L$-Lipschitz and the distance between any $\mathbf{x}$, $\mathbf{x}^\prime$ in the parameter space
$D(\mathbf{x}, \mathbf{x^\prime}) \leq \Delta^2$ for some constant $\Delta$:
\[
R[X] := \sum^T_{t=1} f_t(\widetilde{\mathbf{x}}_t) - f_t(\mathbf{x}^*) = \mathcal{O}\left(\sqrt{T}\right), 
\]
and thus $\lim_{T \to \infty}\frac{R[X]}{T} = 0$.
\end{customthm}

\begin{proof} We follow the proof of~\cite{ssp}.  Define $D(\mathbf{x}, \mathbf{x^\prime}) := \frac{1}{2} \norm{\mathbf{x} - \mathbf{x}^\prime}^2$, where $\norm{\cdot}$ is
the $\ell_2$-norm. Because $f_t$ are convex, we have:
\begin{equation*} 
R[X] := \sum^T_{t=1} f_t(\widetilde{\mathbf{x}}_t) - f_t(\mathbf{x}^*)  \leq \sum^T_{t=1} \left\langle  \widetilde{\mathbf{g}_t},  \widetilde{\mathbf{x}}_t - \mathbf{x}^*\right\rangle.
\end{equation*}
The high level idea is to show that $R[X] = \mathcal{O}\left(\sqrt{T}\right)$, which means $\mathbb{E}_t\left[f_t(\widetilde{\mathbf{x}}_t) - f_t(\mathbf{x}^*)\right] \rightarrow 0$, thus convergence. First, we shall something about the term $\left\langle  \widetilde{\mathbf{g}_t},  \widetilde{\mathbf{x}}_t - \mathbf{x}^*\right\rangle$.
\begin{lem} 
If $X = \mathbb{R}^n$, then for all $t > 0$: 
\begin{align*} 
 \left\langle  \widetilde{\mathbf{x}}_t - \mathbf{x}^*, \widetilde{\mathbf{g}_t} \right\rangle &= \frac{1}{2} \eta_t \norm{\widetilde{\mathbf{g}_t}}^2 
+ \frac{D(\mathbf{x}^*, \mathbf{x}_t)-D(\mathbf{x}^*, \mathbf{x}_{t+1})}{\eta_t} \\
&+ \sum_{p^\prime \neq p} \left[-\sum_{i\in \mathcal{I}_{ p^\prime, c}} \eta_i \left\langle  \widetilde{\mathbf{g}_i},  \widetilde{\mathbf{g}_t}\right\rangle \right].
\end{align*}
\label{lem:1}
\end{lem}
\begin{proof}
\begin{align*} 
D(\mathbf{x}^*&, \mathbf{x}_{t+1})-D(\mathbf{x}^*, \mathbf{x}_t) = \frac{1}{2}\norm{\mathbf{x}^* - \mathbf{x}_{t} + \mathbf{x}_{t} - \mathbf{x}_{t+1}} 
-\frac{1}{2}\norm{\mathbf{x}^* - \mathbf{x}_{t}}  \\
&= \frac{1}{2}\norm{\mathbf{x}^* - \mathbf{x}_{t} + \eta_t \widetilde{\mathbf{g}_t}} 
-\frac{1}{2}\norm{\mathbf{x}^* - \mathbf{x}_{t}} \\
&=  \frac{1}{2} \eta^2_t \norm{\widetilde{\mathbf{g}_t}}^2 - \eta_t \left\langle \mathbf{x}_{t} -  \mathbf{x}^*, \widetilde{\mathbf{g}_t}\right\rangle \\
&=  \frac{1}{2} \eta^2_t \norm{\widetilde{\mathbf{g}_t}}^2  - \eta_t \left\langle \mathbf{x}_{t} -  \widetilde{\mathbf{x}}_t, \widetilde{\mathbf{g}_t}\right\rangle 
- \eta_t \left\langle \widetilde{\mathbf{x}}_t -  \mathbf{x}^*, \widetilde{\mathbf{g}_t}\right\rangle.
\end{align*}
By expanding the second term:
\begin{align*} 
 \left\langle \mathbf{x}_{t} -  \widetilde{\mathbf{x}}_t, \widetilde{\mathbf{g}_t}\right\rangle  &= 
\left\langle  \left[ \sum_{p^\prime \neq p} \sum_{i\in \mathcal{I}_{ p^\prime, c}} \eta_i \widetilde{\mathbf{g}_i}\right],  \widetilde{\mathbf{g}_t}\right\rangle 
= \sum_{p^\prime \neq p} \sum_{i\in \mathcal{I}_{ p^\prime, c}} \eta_i \left\langle \widetilde{\mathbf{g}_i},  \widetilde{\mathbf{g}_t}\right\rangle,
\end{align*}
and moving $\left\langle \widetilde{\mathbf{x}}_t -  \mathbf{x}^*, \widetilde{\mathbf{g}_t}\right\rangle$ to the left, we prove the lemma.
\qed\end{proof}

\noindent Returning to the proof  of the theorem, we use Lemma~\ref{lem:1} to expand the regret $R[X]$:
\begin{align*} 
 R[X] \leq \sum^T_{t=1} \left\langle \widetilde{\mathbf{g}_t},  \widetilde{\mathbf{x}}_t - \mathbf{x}^*\right\rangle = \sum^T_{t=1}  \left( \frac{1}{2} \eta_t \norm{\widetilde{\mathbf{g}_t}}^2 
+  \frac{D(\mathbf{x}^*, \mathbf{x}_t)-D(\mathbf{x}^*, \mathbf{x}_{t+1})}{\eta_t} 
+  \sum_{p^\prime \neq p} \left[-\sum_{i\in \mathcal{I}_{ p^\prime, c}} \eta_i \left\langle  \widetilde{\mathbf{g}_i},  \widetilde{\mathbf{g}_t}\right\rangle \right] \right).
\end{align*}
We now upper-bound each of the terms:
\begin{align}
 \sum^T_{t=1} &\frac{1}{2} \eta_t \norm{\widetilde{\mathbf{g}_t}}^2 \leq \sum^T_{t=1} \frac{1}{2} \eta_t L^2 \hskip 2em (L\mathrm{-Lipschitz~assumption}) \nonumber \\
&= \frac{1}{2} \eta L^2  \sum^T_{t=1} \frac{1}{\sqrt{t}} \leq \eta L^2 \sqrt{T}, \hskip 1em \left(\sum^T_{t=1} \frac{1}{\sqrt{t}} \leq 2 \sqrt{T} \right)
\end{align}
and
\begin{align}
 \sum^T_{t=1} &\frac{D(\mathbf{x}^*, \mathbf{x}_t)-D(\mathbf{x}^*, \mathbf{x}_{t+1})}{\eta_t} = \nonumber \\
&\frac{D(\mathbf{x}^*, \mathbf{x}_1)}{\eta_1} -\frac{D(\mathbf{x}^*, \mathbf{x}_{T+1})}{\eta_T} + \sum^T_{t=2}\left[D(\mathbf{x}^*, \mathbf{x}_t) \left(\frac{1}{\eta_t} - \frac{1}{\eta_{t-1}}\right)\right] \nonumber \\
&\leq \frac{\Delta^2}{\eta} - 0 + \frac{\Delta^2}{\eta} \sum^T_{t=2} \left[\sqrt{t} - \sqrt{t-1}\right] (\mathrm{Bounded~diameter}) \nonumber \\
&= \frac{\Delta^2}{\eta} + \frac{\Delta^2}{\eta} \left[\sqrt{T} - 1 \right] = \frac{\Delta^2}{\eta}\sqrt{T}.
\end{align}
For the last term, we shall use the following lemma:
\begin{lem} 
Let $V$  a normed vector space with norm $\norm{\cdot}$. Let~$\mathbf{y}$ be a vector in $V$ having nonzero entries. For all $\mathbf{x} \in V$, we have 
that $\frac{\norm{\mathbf{x}}}{\norm{\mathbf{y}}} \leq \norm{\mathbf{x} \oslash \mathbf{y}}$, 
where $\oslash$ denotes Hadamard division.
\label{lem:2}
\end{lem}
\begin{proof} By contradiction. That is, suppose that $\frac{\norm{\mathbf{x}}}{\norm{\mathbf{y}}} > \norm{\mathbf{x} \oslash \mathbf{y}}$. Then,
$\norm{\mathbf{x}}  > \norm*{\mathbf{x} \oslash \frac{\mathbf{y}}{\norm{\mathbf{y}}}}$, which is a contradiction since  $\frac{\mathbf{y}}{\norm{\mathbf{y}}}$~is a vector of unit norm.
\qed\end{proof}
We are now in position to upper bound the last term. For simplicity, let $\mathbf{s}_{p^\prime, c}  := -\sum_{i\in \mathcal{I}_{p^\prime, c}} \eta_i \widetilde{\mathbf{g}_i}$ and $\abs{\cdot}$ denote
$\ell_1$-norm or \textit{taxicab} norm Then, 
\begin{align}
 \sum^T_{t=1} &\sum_{p^\prime \neq p} \left[-\sum_{i\in \mathcal{I}_{ p^\prime, c}} \eta_i \left\langle  \widetilde{\mathbf{g}_i},  \widetilde{\mathbf{g}_t}\right\rangle \right]  \nonumber \\
&\leq \sum^T_{t=1} \left(P - 1\right) \left\langle \mathbf{s}_{p^\prime, c}, \widetilde{\mathbf{g}_t} \right\rangle \leq \left(P - 1\right) \sum^T_{t=1}  \left| \left\langle  \mathbf{s}_{p^\prime, c}, \widetilde{\mathbf{g}_t} \right\rangle \right| \nonumber \\
& \leq \left(P - 1\right) \sum^T_{t=1} \norm{\mathbf{s}_{p^\prime, c}} \norm{\mathbf{\widetilde{g}}_t} \hskip 1em (\textrm{Cauchy–Schwarz~inequality}) \nonumber \\
& \leq \left(P - 1\right)L \sum^T_{t=1} \frac{\taxicab{\mathbf{s}_{p^\prime, c}}}{\taxicab{\mathbf{\widetilde{x}}_t}} \taxicab{\mathbf{\widetilde{x}}_t} \hskip 2.5em (L\mathrm{-Lipschitz~assumption}) \nonumber \\
&\leq  \left(P - 1\right)L \sum^T_{t=1} \left\lvert {\mathbf{s}_{p^\prime, c}} \oslash {\mathbf{\widetilde{x}}_t} \right\rvert \taxicab{\mathbf{\widetilde{x}}_t}  \hskip 1em (\textrm{Lemma~\ref{lem:2}}) \nonumber \\
&\leq \left(P - 1\right)L  n \sum^T_{t=1} v_t \taxicab{\mathbf{\widetilde{x}}_t} \hskip 2em (\textrm{Non-significance~of~updates}) \nonumber \\
&\leq  \left(P - 1\right) L n  \sum^T_{t=1} v_t \sqrt{n} \norm{\mathbf{\widetilde{x}}_t} \nonumber \\
\end{align}
\begin{align}
& \leq  \left(P - 1\right) L \left(n \sqrt{n}\right) \sum^T_{t=1} v_t \sqrt{2}\Delta \hskip 2em (\mathrm{Bounded~diameter}) \nonumber \\
& \leq  \sqrt{2}\Delta \left(P - 1\right)  L \left(n \sqrt{n}\right)  \sum^T_{t=1} \frac{v}{\sqrt{t}} \nonumber \\ 
& \leq 2\sqrt{2}\Delta\left(P - 1\right)  L \left(n \sqrt{n}\right) v \sqrt{T}. \hskip 1em \left(\sum^T_{t=1} \frac{1}{\sqrt{t}} \leq 2 \sqrt{T} \right)
\end{align}
\noindent Hence,
\begin{align} 
 R[X] \leq \eta& L^2 \sqrt{T} + \frac{\Delta^2}{\eta}\sqrt{T} \nonumber \\ 
&+ 2\sqrt{2}\Delta\left(P - 1\right)  L \left(n \sqrt{n}\right) v \sqrt{T} = \mathcal{O}\left(\sqrt{T}\right),
\end{align}
\noindent and thus, $\lim_{T \to \infty}\frac{R[X]}{T} = 0$, which concludes the proof.
\qed\end{proof}

\end{document}